\documentclass[aoas]{imsart}
\usepackage{fix-cm}
\usepackage{newtxtext}
\usepackage[inkscapelatex=false]{svg}
\usepackage{float}

\RequirePackage{amsthm,amsmath,amsfonts,amssymb}
\RequirePackage[authoryear]{natbib}
\RequirePackage[colorlinks,citecolor=blue,urlcolor=blue]{hyperref}
\RequirePackage{graphicx}
\RequirePackage{xcolor}
\usepackage{cancel}

\RequirePackage{mathtools}
\usepackage{thm-restate}

\newcommand{\BlackBox}{\rule{1.5ex}{1.5ex}}  % end of proof
\ifdefined\proof
    \renewenvironment{proof}{\par\noindent{\bf Proof\ }}{\hfill\BlackBox\\[2mm]}
\else
    \newenvironment{proof}{\par\noindent{\bf Proof\ }}{\hfill\BlackBox\\[2mm]}
\fi

\usepackage{hyperref}
\hypersetup{
  colorlinks=true,
  filecolor=black,
}

\usepackage[noabbrev,capitalise]{cleveref}
\creflabelformat{equation}{#2\textup{#1}#3}

\usepackage{booktabs} % for professional tables
\graphicspath{{figures/}}

\usepackage{todonotes}
\setuptodonotes{inline}% Set inline as default
\setuptodonotes{color=blue!30} 
\usepackage{silence}
\startlocaldefs
\theoremstyle{plain}
\newtheorem{theorem}{Theorem}

\theoremstyle{remark}
\newtheorem{definition}[theorem]{Definition}

\DeclareMathOperator*{\argmin}{arg\,min}
\newcommand{\ind}{\perp\!\!\!\!\perp}

\endlocaldefs

\makeatletter
\renewcommand\journal@name{}
\makeatother

\begin{document}

\begin{frontmatter}
\newcommand{\mytitle}{Neural Posterior Estimation with Autoregressive Tiling for Detecting Objects in Astronomical Images}
\title{\mytitle}
\runtitle{Autoregressive Tiling for Astronomical Images}
%\thankstext{T1}{A sample of additional note to the title.}

\begin{aug}
\author[A]{\fnms{Jeffrey}~\snm{Regier}\ead[label=e1]{regier@umich.edu}\orcid{0000-0002-1472-5235}}\\
\address[A]{Department of Statistics, University of Michigan\printead[presep={ ,\ }]{e1}}
\end{aug}

\begin{abstract}
Upcoming astronomical surveys will produce petabytes of high-resolution images of the night sky, providing information about billions of stars and galaxies. Detecting and characterizing the astronomical objects in these images is a fundamental task in astronomy---and a challenging one, as most of these objects are faint and many visually overlap with other objects. We propose an amortized variational inference procedure to solve this instance of small-object detection. Our key innovation is a family of spatially autoregressive variational distributions that partition and order the latent space according to a $K$-color checkerboard pattern. By construction, the conditional independencies of this variational family mirror those of the posterior distribution. We fit the variational distribution, which is parameterized by a convolutional neural network, using neural posterior estimation (NPE) to minimize an expectation of the forward KL divergence. Using images from the Sloan Digital Sky Survey, our method achieves state-of-the-art performance. We further demonstrate that the proposed autoregressive structure greatly improves posterior calibration.
\end{abstract}


\end{frontmatter}

%%%%%%%%%%%%%%%%%%%%%%%%%%%%%%%%%%%%%%%%%%%%%%
%%%% Main text entry area:
\section{Introduction}
The next generation of astronomical surveys will produce an unprecedented quantity of high-resolution astronomical images. The Euclid mission, for example, is projected to generate 30 petabytes of image data during its six-year lifespan \citep{romelli2019euclidizing}. Using astronomical images from upcoming surveys, cosmologists will seek to resolve many scientific questions about the nature of the Universe through studies that characterize dark energy and dark matter, map large-scale structure, and search for new types of astronomical objects.

Cataloging astronomical objects (light sources) is a foundational step in analyzing astronomical images. Catalog entries contain information about where an object is located, as well as its type (e.g., star or galaxy), brightness, color, and shape. 
To date, astronomical cataloging has been performed primarily using algorithmic software pipelines with stages that detect objects, segment overlapping objects, and measure properties of interest. 
The stages are generally not based on a coherent statistical model and do not communicate uncertainty information between stages. 
The detection stages rely on peak-finding routines that are applied to the convolution of the image with a matched filter~\citep{lupton2001sdss,morganson2018dark,bosch2018overview}.
This approach does not account for the uncertainty of detection in a coherent way.

In any survey, the majority of objects are faint in terms of signal-to-noise ratio.
Because these faint objects cannot be detected with absolute certainty, the lack of uncertainty quantification in cataloging pipelines is problematic. To date, astronomical analyses have coped with this lack of detection uncertainty using selection criteria (known as ``cuts'') that limit the cohort of studied objects to those with an estimated flux (brightness) in excess of some conservative threshold, thus ensuring so few cataloging errors (both omissions and commissions) within the targeted flux range that the possibility of error can be ignored in downstream analyses \citep[e.g.,][]{strauss2002spectroscopic,blanton2003broadband}.
This strategy is data inefficient because such a conservative flux threshold is required to ensure that the number of cataloging errors is negligible. Moreover, by restricting analyses to only the brightest, most easily detectable sources, we systematically exclude entire classes of astronomical objects that, while individually uncertain, could reveal important insights when studied in aggregate with proper uncertainty quantification.

The scientific cost of these exclusions is substantial.
In weak lensing studies, higher redshift galaxies provide the greatest statistical power for constraining cosmological parameters because they are lensed by more intervening dark matter, and probe earlier cosmic epochs when dark energy constituted a smaller fraction of the Universe's energy density, yet these scientifically valuable galaxies tend to be fainter and are systematically excluded by conservative flux cuts \citep[e.g.,][]{mandelbaum2018weak}. Similarly, these distant galaxies appear ``younger'' because their light reaching us now originated when the Universe was smaller, making them valuable probes of galaxy evolution \citep[e.g.,][]{conselice2014evolution}, but their faintness renders them inaccessible to traditional cataloging approaches.

As telescope hardware improves, selection criteria that ensure a negligible error rate at the level of individual objects become increasingly limiting.
In deeper images, more photoelectrons are recorded in each pixel, so the number of imaged objects (that is, those contributing photoelectrons to the images) increases. As the density of imaged objects increases, a greater number are likely to overlap.
The visual overlap of objects, which is known in astronomy as blending, has become more pronounced as image depth increases. 
% The increase in imaged object density is not offset by improvements in telescope camera resolution because 1) resolution is not increasing as rapidly as depth, and 2) the spatial extent of galaxies and the photoelectron scatter due to the atmosphere are constant.
For example, blending is estimated to occur in 58\% of the imaged objects captured by Subaru Telescope's Hyper Suprime-Cam \citep{bosch2018hyper} and 62\% of the imaged objects captured by the Rubin Observatory Legacy Survey of Space and Time (LSST) \citep{sanchez2021effects}.
With increased blending of objects in survey images, the problem of decomposing an image into objects becomes more ambiguous (i.e., more thoroughly ill-posed), which constitutes a major problem for traditional approaches to cataloging~\citep{melchior2021challenge}.
Contaminated measurements propagate into estimates about populations of stars and galaxies, biasing them. Errors due to blending are expected to be among the leading sources of bias for future surveys.

\paragraph*{Probabilistic cataloging}

Probabilistic cataloging is a potential solution to the problem of blending and detection ambiguity in general. Instead of producing a point estimate of the catalog through algorithmic means, probabilistic catalogs are based on a generative model of images. Given a collection of images, which are modeled as observed random variables, the posterior distribution of the catalog accounts for visually overlapping objects. No special case is required to accommodate blends.

Some existing approaches to probabilistic cataloging have used Markov chain Monte Carlo (MCMC) for approximate Bayesian inference \citep{brewer2013probabilistic, portillo2017improved, feder2020multiband, buchanan2023markov}.
A major limitation of these approaches is that the sampling procedure for MCMC is too slow to process entire modern astronomical surveys \citep{regier2019approximate,liu2023variational}.
Each iteration of an MCMC sampler requires a likelihood evaluation that involves a large number of pixels (i.e., petabytes for the volume of data generated in upcoming surveys). Even with clever sampler strategies, adequate mixing can require more iterations than are feasible.

Variational inference (VI) can be a computationally efficient alternative to MCMC \citep{blei2017variational, zhang2018advances}. VI uses numerical optimization to find a posterior approximation from a class of candidate approximations.
One challenge in using variational inference for object detection is that the variational objective can be difficult to optimize: the gradient of the posterior density gives little or no guidance in how to improve the positions of inferred objects (according to the current iterate) that are not sufficiently close to an imaged object \citep{liu2023variational}.
Another challenge in using variational inference for object detection is that the posterior is transdimensional because the number of objects is ambiguous, and a reasonable variational distribution requires the flexibility to model this ambiguity.

\paragraph*{Probabilistic cataloging with tiling}

\citet{liu2023variational} proposed a VI method for detecting stars, called StarNet, which in many respects overcomes both challenges. StarNet decomposes the problem of detecting objects in large images into subproblems. Each subproblem involves performing posterior inference for the objects centered in a small region of the sky, called a \textit{tile}, based on a relevant subimage.
\citet{liu2023variational} defines tiles, using the coordinates of the image, to be disjoint $2\times 2$-pixel regions of sky. The subimage for each tile-based subproblem is a $32\times 32$-pixel image centered on the tile. (Objects typically extend far beyond the boundaries of the tile in which they are centered.) Whereas the tiles form a partition of the sky, their subimages overlap with each other.

In StarNet, these tile-based inference subproblems are solved efficiently through amortized variational inference, with a variational distribution that factorizes fully over tiles: the objects in each tile are independent under the variational distribution. 
The variational distribution is parameterized by a neural network, which maps each subimage to the distributional parameters for a tile.
For each tile, the variational distribution has support for only a small number of objects (e.g., four). This is not a major limitation given the small tile size, yet it greatly simplifies transdimensional inference.

\paragraph*{The problem with independent tiling}
For object-detection models, in the exact posterior, nearby objects are often highly dependent.
For example, the detection of an object will typically have a suppressive effect (negative dependence) on the detection of a nearby object: because the former detection ``explains away'' a bright region of an image, a second nearby object is less likely. Therefore, StarNet's assumption of independence in the variational distribution among neighboring tiles is limiting.

This limitation can be quite severe in practice. Consider a variational distribution with 2$\times$2-pixel tiles. If objects are uniformly distributed spatially, then 19\% of objects are expected to be centered within 0.2 pixels of the border. 
For an object within 0.2 pixels of a tile border, it is highly ambiguous which tile contains its centroid (see \cref{fig:moving_star,fig:border-experiment} for examples).
With independent tiling, in the vicinity of this object, many samples from the variational distribution would therefore incorrectly include either two objects or no objects.
Even if this miscalibration of the variational distribution affects a small percentage of objects, it can invalidate downstream analyses, which typically consider populations of objects in aggregate.

This problem is not solved by simply increasing the tile size because in StarNet, the objects detected within a single tile are also treated as independent under the variational distribution given the total number of objects in the tile.
Further, as tile size grows, the subproblems become more difficult to solve, negating the benefit of decomposing the original full-image inference problem in the first place.

Despite this limitation, StarNet outperformed a well-tuned MCMC sampler both in terms of accuracy and computational efficiency, suggesting that a tiling-based approach has important advantages over alternative approaches.

\paragraph*{Our contributions}
First, we formalize the inference problem addressed by tiling, as previous work did not rigorously connect the solution to the original inference problem (i.e., inferring the posterior given a large image) with the solutions to the tiling subproblems, which have support for a bounded number of objects. Using the thinning property of the Poisson distribution, we recast our generative model, which is a marked spatial Poisson process (\cref{sec:model}), as a distribution that fully factorizes over tiles. With this new representation of the latent space, the solutions to the tiling subproblems can be shown to constitute a particularly important marginal of the posterior distribution (\cref{sec:change-of-variables}).

Second, we propose a novel spatially autoregressive family of variational distributions (\cref{sec:variational_family}).
We assign each tile a ``rank'' according to a $K$-color checkerboard pattern so that neighboring tiles have different ranks. Our variational family factorizes spatially over these ranks, which induce a partial ordering of the tiles.
%This variational family allows inference by sampling in parallel over regions that are sufficiently far apart, while sampling sequentially locally.
The variational family is amortized, with distributional parameters produced by a fully convolutional inference network. If the receptive fields of this network are set optimally, then, by construction, the conditional independence assertions of the variational family mirror those of the generative model. This property of the variational distribution makes it flexible enough to approximate the posterior well, but not more flexible than is required to do so, so that training can focus on aspects of the posterior that are not known a priori.

Third, we consider neural posterior estimation (NPE) as an alternative to traditional variational inference (\cref{sec:fitting}).
Whereas traditional VI maximizes the evidence lower bound (ELBO), and therefore minimizes the Kullback-Leibler (KL) divergence from the variational distribution to the exact posterior, NPE maximizes the information lower bound derived in \citet{barber2004im}, which involves maximizing a particular expectation of the KL divergence from the exact posterior to its approximation. (Recall that the KL divergence is not symmetric.)
NPE is equivalent to the reweighted wake-sleep algorithm \citep{bornschein2015reweighted} performed without wake-phase updates (i.e., with only sleep-phase updates). It has also been termed ``forward amortized variational inference''~\citep{ambrogioni2019forward}.
In addition, it has been proposed for likelihood-free inference (a.k.a. simulation-based inference) by \cite{papamakarios2016fast}.
% NPE has attractive properties that have not been widely recognized, as well as drawbacks that have not been realized.

Fourth, we conduct two extensive case studies to assess the empirical performance of the proposed method (\cref{sec:experiments}).
The first focuses on the analysis of high-resolution images of both stars and galaxies, while the second focuses on the analysis of an image of a starfield with an extremely high density of objects. In both, the proposed method outperforms the existing methods. Through ablation studies, we probe the extent to which the spatially autoregressive aspect of the proposed method contributes to calibration and performance.

We conclude by discussing a surprising limitation of NPE, which we connect to the issue of ``exposure bias'' in autoregressive text generation (\cref{sec:discussion}). We also discuss several limitations of our tile-based variational family and the relevance of our work to astronomical practice. There are barriers to the widespread adoption of our approach in the astronomy community, including the need to develop methods that propagate object-level uncertainty information to population-level estimates  (\cref{sec:asto-discussion}). There is also great potential, both to improve the accuracy of astronomical catalogs and to put astronomical image analysis on a sounder statistical footing.

%%%%%%%%%%%%%%%%%%%%%%%%%%%%%%%%%%%%%%%%%%%%%%%%%%%%%%%%%%%%
\section{Generative models of astronomical images} \label{sec:model}
%%%%%%%%%%%%%%%%%%%%%%%%%%%%%%%%%%%%%%%%%%%%%%%%%%%%%%%%%%%%
Astronomical images record the number of photoelectrons detected at each pixel of a charge-coupled device (CCD) image sensor.
This count depends on multiple factors, including the number of astronomical objects in the image; these objects' positions, types, spectral energy densities, and shapes; the background intensity and gain; and the point spread function. Multiple objects often contribute photoelectrons to the same pixel.

In this section, for concreteness, we introduce a specific statistical model that is similar to the models proposed in previous work on probabilistic detection of astronomical objects \citep{portillo2017improved, feder2020multiband, liu2023variational, buchanan2023markov}.
However, our model is more explicit about the spatial extent (size) of imaged objects because explicit bounds on spatial extent are useful for reasoning about the posterior dependence structure.

%%%%%%%%%%%%%%%%%%%%%%%%%%%%%%%%%%%%%%%%%%%%%%%%%%%%%%
\subsection{Prior} \label{sec:prior}
%%%%%%%%%%%%%%%%%%%%%%%%%%%%%%%%%%%%%%%%%%%%%%%%%%%%%%
The prior on catalogs is a marked spatially homogeneous Poisson process. Objects appear in an image with rate $\mu$. For an $H\times{}W$-pixel image, the number of objects in this image is
\begin{align}
    \label{eq:n_prior}
    S &\sim \text{Poisson}(\mu HW).
\end{align} 
For $s=1,\ldots,S$, the position $u_s$ of object $s$ is distributed uniformly within the image:
\begin{align*}
  u_s \mid S & \sim \text{Uniform}([0, H] \times [0, W]).
\end{align*}
Each object is additionally marked by properties such as object type (e.g., star, galaxy), flux (brightness) in various frequency ranges, and, for galaxies, morphological characteristics.
We refer to these properties of object $s$ collectively as $v_s$, and
\begin{align*}
    v_s \mid S, u_s \sim F(u_s),
\end{align*}
where $F(u_s)$ is a distribution, optionally parameterized by position $u_s$, that is sampled independently for each imaged object.
Because our focus is object detection rather than object characterization, the specific form of $F$ is not important for what follows, provided that $F$ can be sampled efficiently.

%%%%%%%%%%%%%%%%%%%%%%%%%%%%%%%%%%%%%%%%%%%%%%%%%%%%%%
\subsection{Conditional likelihood} \label{sec:likelihood}
%%%%%%%%%%%%%%%%%%%%%%%%%%%%%%%%%%%%%%%%%%%%%%%%%%%%%%
Let $y = \{u_1, v_1,\ldots, u_S, v_S \}$ denote a latent catalog that is sampled from the prior.
Then, for $h=1,\ldots, H$ and $w=1,\ldots,W$, the photoelectron count for pixel $(h, w)$ is
\begin{align*}
    x_{h,w} \mid y \sim \mathrm{Poisson}(\lambda_{h,w}(y)),
\end{align*}
where
\begin{align*}
\lambda_{h,w}(y) \coloneqq \gamma_{h,w} + \sum_{s=1}^s \xi_{h,w}(u_s, v_s).
\end{align*}
Here, $\gamma_{h,w}$ is the expected background intensity for pixel $(h,w)$ and $\xi_{h,w}(\cdot)$ is a nonnegative function that maps an object's properties to its expected contribution to pixel $(h, w)$. In this work, for simplicity, both quantities are taken to be deterministic and known, with values supplied by the existing survey pipelines. If it were desirable to treat them as random, the method of \citet{patel2025neural}, which is complementary to this work, could be used to train an inference network with random backgrounds and random point spread functions.

We require $\xi_{h,w}(\cdot)$ to have the property that $\xi_{h,w}(u_s, v_s) = 0$ if $\|[h, w]^\top - u_s\|_\infty > R$, for some $R > 0$. We refer to $R$, which encodes the locality of light sources, as the maximum object radius measured in pixels.

An implication of the additive form of $\lambda_{h,w}(y)$ is that there is no occlusion among astronomical objects---an assumption that is well-founded for most applications.
Because the number of arrivals is large, it is common to use the normal approximation to the Poisson distribution:
\begin{align*}
  x_{h, w} \mid y \sim \mathrm{Normal}(\lambda_{h,w}(y),\, \lambda_{h,w}(y)).
\end{align*}

The method that follows applies equally well regardless of whether the number of arrivals is modeled as Poisson or Gaussian. We refer to the pixels of an image collectively as $x \coloneqq (x_{1,1}, \ldots, x_{H,W})$.

% The precise form of the likelihood is not of great concern in what follows because the proposed inference method is ``likelihood-free'': it does not require explicit evaluation of the likelihood.  The ability to efficiently sample the generative model is all that is required for inference.

\section{A new representation of catalogs}
\label{sec:change-of-variables}

In the model described in the previous section, a latent astronomical catalog is represented as a random set $y$, which is a complex probabilistic structure: it is transdimensional, it does not provide variable names that allow for indexing objects in a spatially coherent way, and it does not distinguish between objects that are of interest (e.g., because they are bright enough to detect) and those that are not.
In this section, we propose a new representation of a catalog that allows us to more explicitly notate the spatial proximity of objects (\cref{sec:spatial}) and objects that are part of a cohort of interest (\cref{sec:of-interest}).

\subsection{Tiling}
\label{sec:spatial}
To facilitate reasoning about spatial proximity among objects, we perform a change-of-variable operation on the latent variables, given by the following bijection:
\begin{align}
\label{eq:tiling}
f_1(y) \coloneqq
\begin{bmatrix}
    y_{1,1} & \cdots & y_{1, W'} \\
    \vdots  & \ddots & \vdots    \\
    y_{H',1} & \cdots & y_{H', W'}
\end{bmatrix},
\end{align}
where $H' = H / T$, $W' = W / T$, and $T$ is a user-defined tile side length in pixels. In our experiments, $T = 2$ or $T = 4$.

The entries in this matrix are sets. For row $h$ and column $w$, partial catalog $y_{h,w}$ is the subset of astronomical objects in $y$ whose centers have pixel coordinates in $[(h - 1)T, hT) \times [(w - 1)T, wT)$.

We denote the set of all tile indices $(h, w)$ as $\Omega \coloneqq \{1,\ldots,H'\}\times\{1,\ldots,W'\}$.
For a subset of the tile indices $L \subset \Omega$, we denote the corresponding latent random variables $y_L \coloneqq \{y_i : i \in L\}$.
We refer to the elements of $y_\Omega$ as the tile latent variables.

Posterior independence between pairs of tile latent variables, even those corresponding to distant tiles, does not generally hold for this model: random variables are coupled through the chains of potential objects between them. However, conditional independence does hold for tile latent variables that are sufficiently far apart, given the other tile latent variables, as shown in Appendix A.
We subsequently develop a family of candidate posterior approximations with this property enforced by construction rather than having to learn through potentially costly and inexact numerical optimization.

\subsection{Objects of interest}
\label{sec:of-interest}

Often in astronomical object detection, only some objects are of interest. In any region of the sky, the vast majority of objects are so faint that they do not contribute even a single photoelectron to any particular image. Other objects contribute a small number of photoelectrons to images (e.g., tens); their contribution is great enough to matter in aggregate, but not so great that any of these objects can be identified and localized with sufficient precision to be useful.
Thus, it is common practice in astronomical analyses to stratify objects by flux and/or redshift and focus inquiries on a cohort of them.

A change of variables gives us notation to describe inference for the objects of interest, without ignoring the aggregate influence of the remaining objects. In particular, for each tile latent variable, we set
$f_2( y_\ell ) \coloneqq [z_\ell, e_\ell]^\top$,
where, for the objects centered in the tile $\ell$, the objects of interest are denoted $z_\ell$, and others, which we refer to as ``nuisance objects,'' are denoted $e_\ell \coloneqq y_\ell \setminus z_\ell$.
In our experiments, we restrict the objects of interest to those with flux in excess of a threshold. Additionally, if this set includes more than $M$ objects centered in a particular tile, we further restrict the objects of interest to the $M$ brightest objects in this tile, where $M$ is a user-defined integer.

Observe that $\{z_{1,1}, e_{1,1}, \ldots,z_{H',W'}, e_{H',W'}\}$ is a partition of $y$.
Collectively, we refer to objects of interest as $z \coloneqq \{z_{1,1}, \ldots,z_{H',W'}\}$ and nuisance objects as $e \coloneqq \{e_{1,1}, e_{1,1}, \ldots,e_{H',W'}\}$.
We aim to infer $p(z \mid x)$,
which is the marginal of the posterior distribution that describes the objects of interest.
This quantity is potentially easier to infer than the full posterior, as the number of objects is bounded under it.
Restricting our attention to this marginal of the posterior does not change the generative model: nuisance objects are nevertheless modeled as contributing photoelectrons to the observed image.

% Remark: Move the number of objects per tile, $N_\ell$, to $z'_\ell$ if truncation is a concern
% could say something about the amount of filtering likely in practice, e.g., for sdss m2

\begin{table}[p]
\centering
\small
\caption{Summary of mathematical notation used throughout the manuscript, ordered by first appearance.
}
\label{tab:notation}
\begin{tabular}{ll|ll}
\toprule
\textbf{Symbol} & \textbf{Description} & \textbf{Symbol} & \textbf{Description} \\
\midrule
\multicolumn{4}{c}{\textit{\textbf{Section 2: Generative model}}} \\
\midrule
$\mu$ & Density of astronomical objects & $u_s$ & Position of object $s$ \\
$H$ & Image height in pixels & $v_s$ & Properties of object $s$ \\
$W$ & Image width in pixels & $F$ & Distribution of object properties \\
$S$ & Number of objects in image & $y$ & Latent catalog (all objects) \\
$s$ & Object index & $h, w$ & Pixel row/column indices \\
$x_{h,w}$ & Photoelectrons observed at $(h,w)$ & $\gamma_{h,w}$ & Background intensity \\
$\lambda_{h,w}(y)$ & Expected intensity at $(h,w)$ & $\xi_{h,w}(\cdot)$ & Object's contribution to pixel \\
$R$ & Maximum object radius & $x$ & Full astronomical image \\
\midrule
\multicolumn{4}{c}{\textit{\textbf{Section 3: A new representation of catalogs}}} \\
\midrule
$f_1$ & Tiling transformation & $\Omega$ & Set of all tile indices \\
$T$ & Tile side length & $L$ & Subset of tile indices \\
$H'$ & Number of tile rows ($H/T$) & $y_L$ & Tile variables for tiles in $L$ \\
$W'$ & Number of tile columns ($W/T$) & $y_\Omega$ & All tile variables \\
$y_{h,w}$ & Objects in tile $(h,w)$ & $f_2$ & Objects of interest transform \\
$\ell$ & Tile index, $\ell \coloneq (h, w)$ & $e$ & All nuisance objects \\
$y_\ell$ & Objects in tile $\ell$ & $M$ & Max objects of interest per tile \\
$z_\ell$ & Objects of interest in tile $\ell$ & $z$ & All objects of interest \\
$e_\ell$ & Nuisance objects in tile $\ell$ & & \\
\midrule
\multicolumn{4}{c}{\textit{\textbf{Section 4: A spatially autoregressive variational family}}} \\
\midrule
$p(z \mid x)$ & Posterior (target) distribution & $z_\ell^{[i]}$ & $i$-th object in tile $\ell$ ($\ell \in \mathbb N^2)$ \\
$d$ & Divergence measure & $z_k^{[i]}$ & $i$-th objects in rank-$k$ tiles ($k \in \mathbb N^1)$ \\
$Q$ & Variational family & $z_k^{[<i]}$ & Pre-$i$ objects in rank-$k$ tiles \\
$q$ & Variational distribution & $f$ & Inference network \\
$q^\star$ & Optimal variational dist. & $f_\mathcal{X}$ & Image backbone network \\
$K$ & Number of checkerboard ranks & $f_\mathcal{N}$ & Neighborhood network \\
$C_k$ & Set of rank-$k$ tiles & $f_\mathcal{D}$ & Detection head network \\
$\Psi(h,w)$ & Rank function for tile & $r_{\mathcal{X}}$ & Image receptive field radius \\
$z_{<k}$ & Pre-rank-$k$ tile variables & $r_{\mathcal{N}}$ & Neighborhood receptive radius \\
$z_k$ & Rank-$k$ tile variables & $D_1$ & Neighborhood channels \\
$\sigma$ & Permutation of objects & $D_2$ & Embedding dimension \\
\midrule
\multicolumn{4}{c}{\textit{\textbf{Section 5: Fitting the variational distribution}}} \\
\midrule
$\phi$ & Variational parameters & $q_\phi$ & Parameterized var. dist. \\
$\mathcal{L}_{\text{rev}}(\phi)$ & Reverse KL objective & $\mathcal{L}_{\text{fwd}}(\phi)$ & Forward KL objective \\
$c$ & Constant w.r.t. $\phi$ & & \\
\bottomrule
\end{tabular}
\end{table}

%%%%%%%%%%%%%%%%%%%%%%%%%%%%%%%%%%%%%%%%%%%%%%%%%%%%%%
\section{A variational family with spatially autoregressive tiling} \label{sec:variational_family}
%%%%%%%%%%%%%%%%%%%%%%%%%%%%%%%%%%%%%%%%%%%%%%%%%%%%%%
Given an astronomical image $x$, we wish to infer the posterior distribution $p(z \mid x)$ of catalogs $z$.
Since images often contain hundreds or thousands of astronomical objects, each with continuous parameters (e.g., position, flux), the posterior is defined over a space that is both high-dimensional and transdimensional~\citep{brewer2013probabilistic,portillo2017improved}.
Exact sampling from the posterior and computing expectations with respect to it are intractable in this setting. We therefore turn to variational inference (VI) to approximate the posterior with a tractable distribution.

VI recasts posterior inference as numerical optimization, which facilitates computationally efficient solutions \citep{blei2017variational, zhang2018advances}.
Given observations $x$, VI selects an approximation to the posterior $p(z \mid x)$ from a tractable family of candidate distributions $Q$, known as a variational family. This is achieved by finding a distribution $q^\star \in Q$ that minimizes a divergence $d$:
\begin{align}
    \label{eq:vi}
    q^\star = \argmin_{q \in Q} d(p(z \mid x), q(z)).
\end{align}
The remainder of \cref{sec:variational_family} proposes a particular variational family $Q$ for use in \cref{eq:vi}. We present the checkerboard-based factorization of the variational distribution (\cref{sec:vf_form}), its model of objects within the same tile (\cref{multiple-objects}), and the architecture of a neural network that enables amortization (\cref{sec:amortized}). The proposed variational family is autoregressive in that tiles have a sequential dependence structure and dependencies are spatially local. The proposed variational family is not merely more flexible than that of its predecessor (StarNet):
its dependence structure precisely reflects that of the posterior (\cref{sec:mirrors}).

\Cref{tab:notation} provides a summary of the mathematical notation.

\subsection{A checkerboard-based factorization}
\label{sec:vf_form}

To define a variational family $q$, we first define a partial ordering over tiles. In $q$, the latent variables associated with each tile depend only on the latent variables of tiles that  come strictly before it in this ordering.

\begin{figure}
    \centering
    \includegraphics[width=0.3\textwidth]{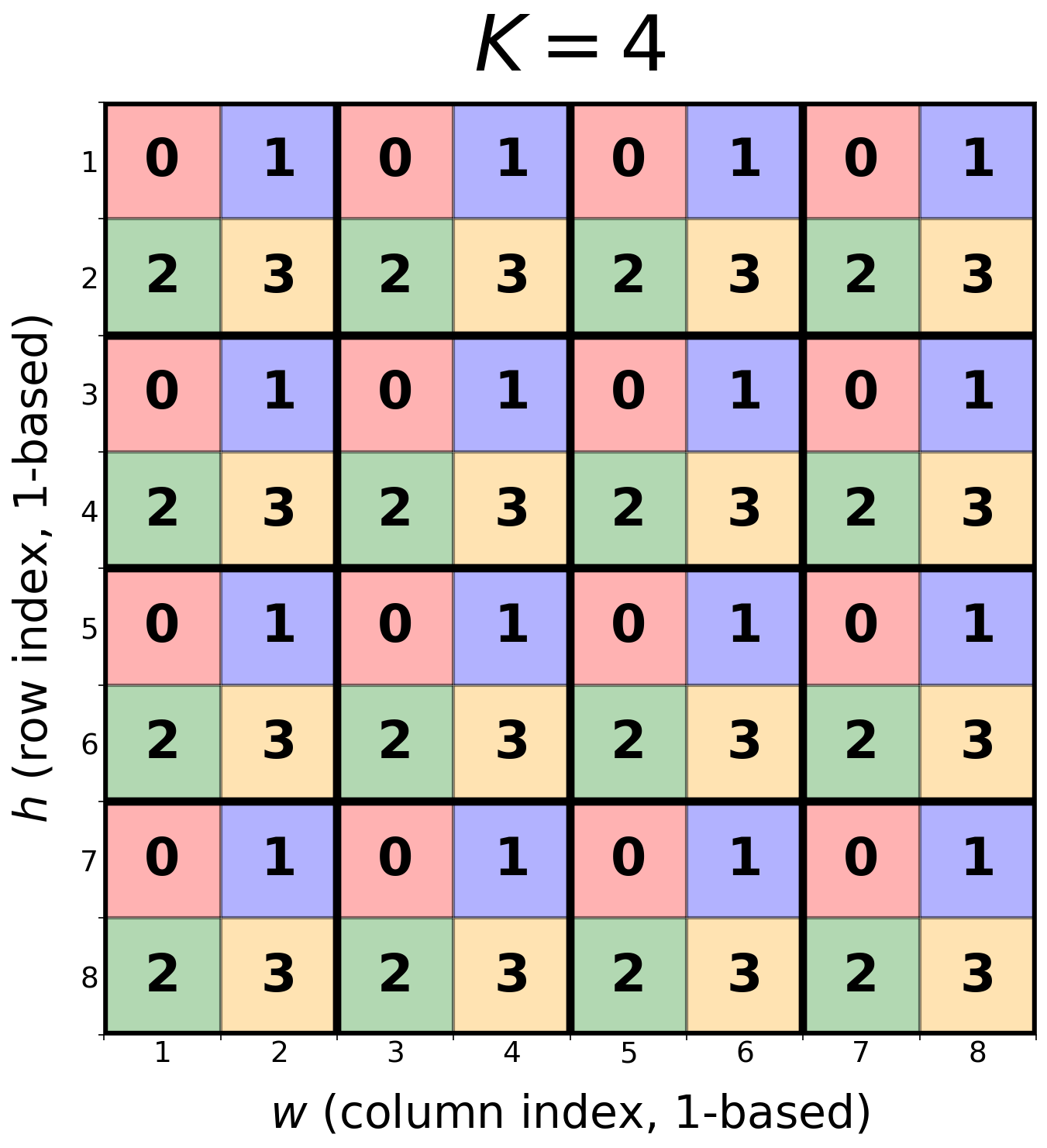}
    \hspace{1em}
    \includegraphics[width=0.3\textwidth]{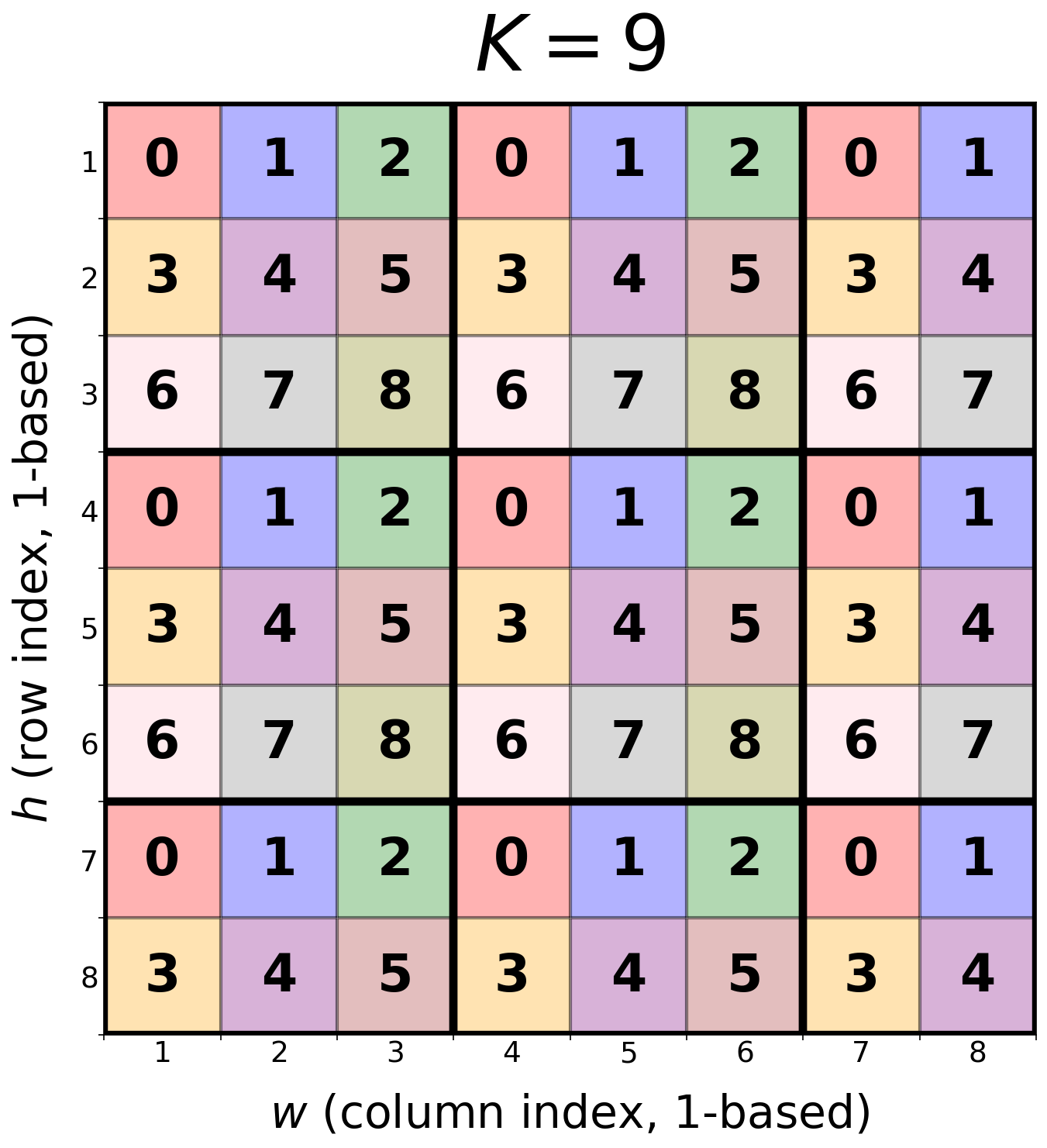}
    \hspace{1em}
    \includegraphics[width=0.3\textwidth]{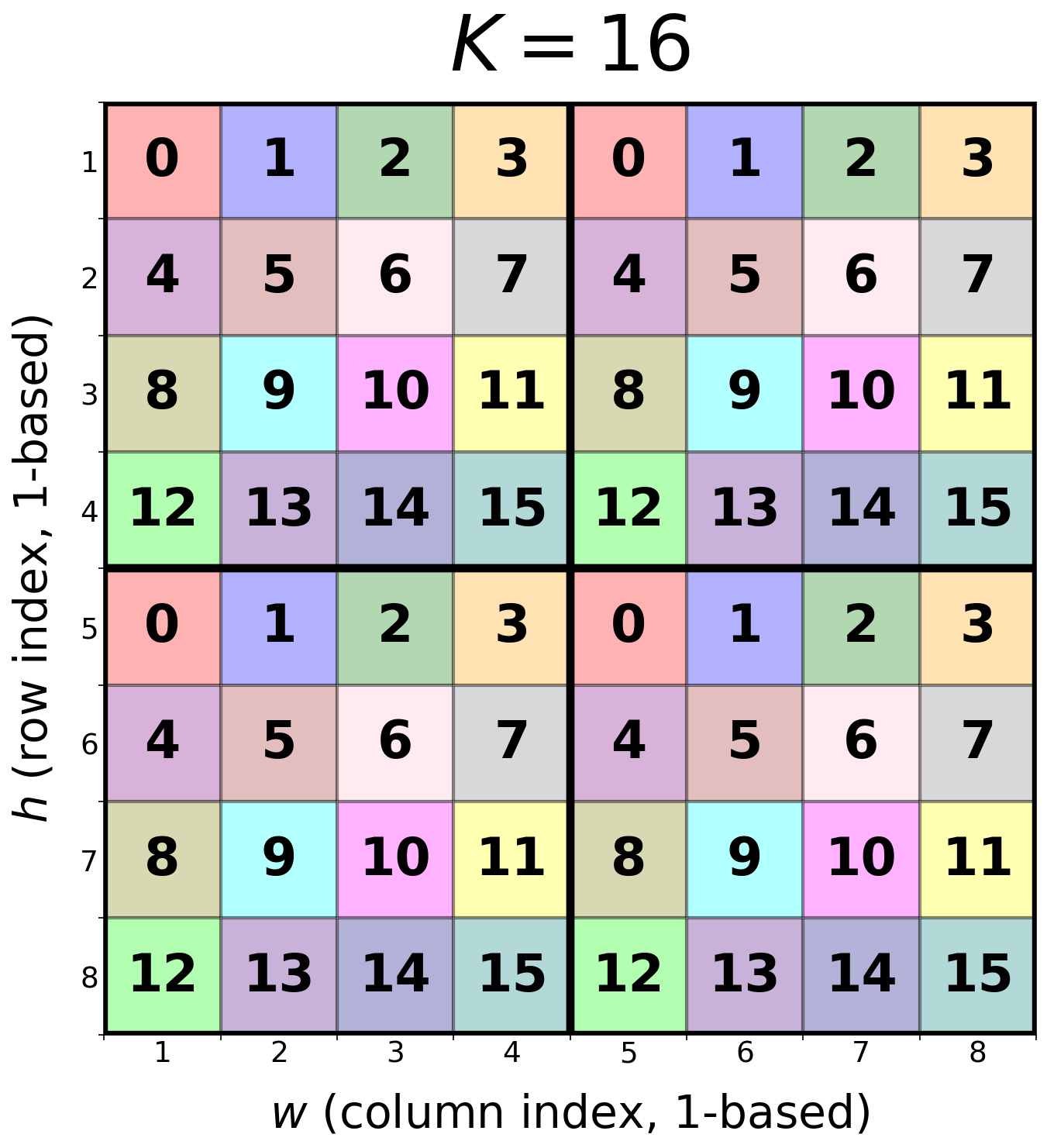}
    \caption{$K$-color checkerboard patterns for $K=$ 4, 9, and 16 colors, with color-specific ranks assigned by $\Psi(h, w)$. Tiles of the same color are separated by 1, 2, and 3 tiles, respectively, in these checkerboards, and $\sqrt{K} - 1$ in general.}
    \label{fig:checkerboards}
\end{figure}

To define a partial ordering of tiles, let $K$ be an integer that is a perfect square. 
For $h=1,\ldots,H'$ and $w=1,\ldots,W'$, let 
\begin{align*}
    \Psi(h, w) \coloneqq ((h - 1) \bmod \sqrt{K}) \cdot \sqrt{K} + ((w - 1) \bmod \sqrt{K}).
\end{align*}
We refer to $\Psi(h, w)$ as the rank of the tile with index $(h, w)$.
For $k = 1,\ldots,K$, let the rank-$k$ tiles be
\begin{align*}
C_k \coloneqq \{(h,w) \in \Omega : \Psi(h, w) = k\}.
\end{align*}
Then $\{C_1, \ldots, C_K\}$ is a partition of $\Omega$ that follows a $K$-color ``checkerboard'' pattern. \cref{fig:checkerboards} shows several checkerboard patterns of this form.

For any tile index $\ell = (h, w)$, let $z_\ell \coloneqq z_{h,w}$.
For $k=1,\ldots,K$, let the pre-rank-$k$ tile variables $z_{<k} \coloneqq \left\{z_\ell : \Psi(\ell) < k \right\}$ and let the rank-$k$ tile variables $z_{k} \coloneqq \left\{z_\ell : \Psi(\ell) = k \right\}$.
Note the overloading of notation here: $z_\ell$, with $\ell$ a length-two vector, refers to the latent variables associated with one tile whereas $z_k$, with $k$ a natural number, refers to the latent variables associated with a collection of tiles.

We approximate $p(z \mid x)$ with a variational distribution that factorizes over ranks:
\begin{align}
    q(z \mid x) \coloneqq \prod_{k = 1}^K q(z_k \mid z_{<k}, x). \label{eq:q_ranks}
\end{align}
% The factorization turns the joint modeling problem into a sequence problem, where one learns to predict the next pixel given all the previously generated pixels (pixelrnn)

\subsection{The within-rank variational distribution}
\label{multiple-objects}
Recall that each tile contains at most $M$ objects of interest.
Let $\sigma$ be a permutation of $(1, \ldots, M)$, which induces an ordering on the set of objects of interest.
In the context of a particular ordering $\sigma$, for a tile with coordinates $\ell \in \Omega$, denote the elements of $z_\ell$ as $z_\ell^{[1]}, \ldots, z_\ell^{[M]}$, with the convention that if $|z_\ell| < M$, it is padded with a (non-random) null value. (The dependence on $\sigma$ is implicit.) 
For a particular rank $k$, let the $i$th objects $z_k^{[i]} \coloneqq \left\{ z_\ell^{[i]} : \ell \in C_k \right\}$
and let the pre-$i$ objects $z_k^{[<i]} \coloneqq \left\{ z_k^{[1]}, \ldots, z_k^{[i - 1]} \right\}$.

In our variational distribution, the ordering $\sigma$ is an auxiliary random variable; it is not part of the generative model. We define $q$ over an extended space of latent variables that includes both the latent variables of interest $z_k$ and the permutation $\sigma$. This notational convention of reusing $q$ to denote the distribution over auxiliary variables follows convention in the variational inference literature \citep[e.g.,][]{maaloe2016auxiliary}.

As $\sigma$ is an auxiliary variable, the posterior approximation under $q$ is obtained by marginalizing $\sigma$. For a generic functional $f$, we denote this marginalization as
\begin{align}
    \mathbb E_{\sigma \sim q(\sigma)} f(\sigma) \coloneqq \sum_\sigma f(\sigma) q(\sigma).
\end{align}
Defining the posterior approximation as the result of this marginalization ensures that the approximation is a well-defined probability distribution that integrates to one, whereas this property would not be as apparent if it were instead defined directly as a weighted sum.

Under $q$, the permutation $\sigma$ marginally follows a uniform distribution over the $M!$ possible permutations. This represents a modeling choice made for simplicity. Alternatively, for example, when converting a set of astronomical objects to a sequence, we could have upweighted mappings that sort objects by flux. Such domain-specific ordering heuristics might offer computational advantages, such as more stable training or faster convergence. However, we choose the uniform distribution over permutations because it requires no additional hyperparameters or domain knowledge about optimal orderings.

We return now to \cref{eq:q_ranks}, to expand the right-hand side.
For $k=1,\ldots,K$, we set
\begin{align}
    \label{eq:rank-k-tiles-marginal-order}
    q(z_k \mid z_{<k}, x) \coloneqq \mathbb E_{\sigma \sim q(\sigma)} \left[ q(z_k \mid \sigma, z_{<k}, x) \right],
\end{align}
which can be computed exactly by summing over the $M!$ permutations. This sum can be calculated quickly if $M$ is small (e.g., less than 5), which is the setting we have in mind. \cref{sec:gp-discussion} discusses approximations for larger $M$.

We additionally set
\begin{align}
    \label{eq:rank-k-tiles-given-order}
    q(z_k \mid \sigma, z_{<k}, x) \coloneqq \prod_{i = 1}^M q(z_k^{[i]} \mid z_k^{[<i]}, z_{<k}, x).
\end{align}
Here, the right-hand side depends implicitly on $\sigma$ to order the objects.
Finally, we specify that
\begin{align}
q(z_k^{[i]} \mid z_k^{[<i]}, z_{<k}, x)
\coloneqq q\left(z_k^{[i]} \mid f\left(z_k^{[<i]}, z_{<k}, x\right) \right),
\end{align}
where $f$ is a neural network, whose structure we explain next.

\subsection{The inference network}
\label{sec:amortized}
So far we have described partial catalogs (e.g., $z_{<k}$) in mathematical terms as collections of sets organized according to an $H' \times W'$ grid with some cells excluded. To implement our variational distribution using convolutional neural networks (CNNs), we must address a practical constraint: CNNs operate on fixed-size tensors, not variable subsets of data. This constraint is fundamental to leveraging the parallel processing capabilities of GPUs.

We handle this using binary masking, a common technique in neural network design \citep[e.g.,][]{devlin2018bert}. For any partial catalog, we represent it as a complete $H' \times W'$ tensor where missing entries are set to zero, accompanied by a binary mask tensor of the same dimensions. The binary mask indicates which entries are present. When a masked catalog serves as network input, we concatenate both the masked data tensor and the mask itself, allowing the network to distinguish between entries that are truly zero-valued and those that are absent. This representation maintains the fixed tensor dimensions required for CNN operations.

In this section, when partial catalogs (e.g., $z_{<k}$) appear as inputs or outputs of our neural networks, they should be understood as using this masked representation. The neural network $f$ has the following form:
\begin{align*}
f\left(z_k^{[<i]}, z_{<k}, x \right)
\coloneqq f_\mathcal{D}\left( f_\mathcal{X}(x),\, f_\mathcal{N}\left(z_{<k}, z_k^{[<i]}\right) \right).
\end{align*}
The network $f$, represented graphically in \cref{fig:architecture}, consists of three components: the image backbone $f_\mathcal{X}$, the neighborhood network $f_\mathcal{N}$, and the detection head $f_\mathcal{D}$. We give more detail about each of these three component networks below, and in Appendix~D we specify the precise network architectures.

Note that by parameterizing $f$ as a neural network, rather than using an iterative optimization procedure as in traditional VI, we achieve amortization: the same network can be applied to arbitrary images and partial catalogs without retraining. Additionally, our network architecture, which has a limited receptive field with respect to both the image and the partial catalogs, provides locality in the dependence structure of our variational family.

\begin{figure}
    \centering
    \includegraphics[width=0.55\linewidth]{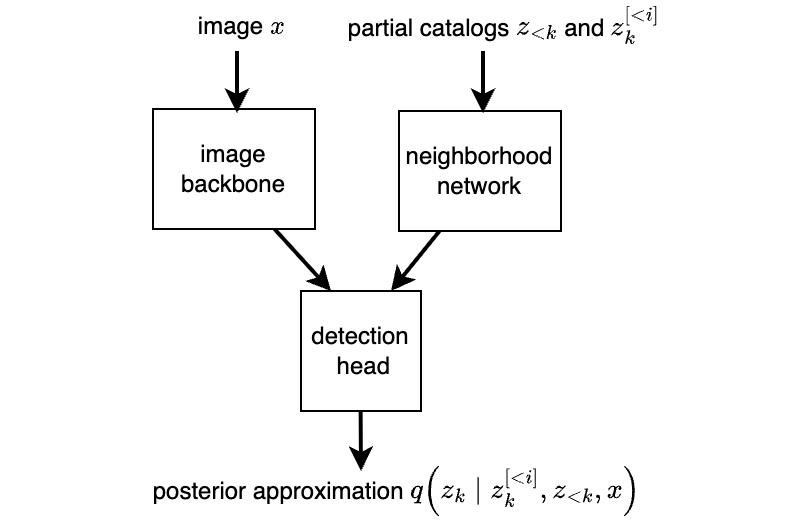}
    \caption{The inference network. Sampling from $q(z \mid x)$ or evaluating a likelihood under it involves a single forward pass of the image backbone, whereas the neighborhood network and detection head must be called iteratively, for each tile rank $k \in \{1,\ldots K\}$ and each potential object $i \in \{1,\ldots,M\}$. 
    % Likelihood evaluation, but not sampling, additionally requires iterating over the $M!$ possible permuations, each represented by $\sigma$.
    }
    \label{fig:architecture}
\end{figure}

%  The variational parameters of $q$ are the weights of an inference network, which maps overlapping subimages centered on each tile to the distributional parameters of $q$.
%Because this mapping is performed at each spatial position of an image, our inference network can be efficiently implemented as a fully convolutional network.

\paragraph*{The image backbone}
The image backbone, $f_\mathcal{X}$, takes the image as input and outputs features that are useful for downstream inference.
It contains 4.9 million parameters, the majority of the parameters of the inference network. The backbone architecture resembles that of a U-net, which is a CNN composed of a sequence of downsampling layers followed by upsampling layers, with down- and up-blocks of the same spatial resolution linked through skip connections \citep{ronneberger2015u}.

We additionally adopt aspects of the architecture of YOLOv5, which is a real-time, single-stage object detection model, known for innovations such as cross-stage partial blocks \citep{redmon2016you}.
In contrast to YOLOv5, we use group normalization instead of batch normalization because we found the latter to work poorly in the context of a recurrent architecture. Additionally, replacing batch norm with group norm allows us to train with small batches without instability, thus conserving GPU memory.

In the image backbone, the degree of downsampling typically exceeds the degree of upsampling, resulting in a spatially coarser output than the image input. The degree of coarsening must accord with the size of the tile $T$. 
For example, if the tile size is 4$\times$4 pixels, our network must downsample two times more than it upsamples, with a stride of two. If we think of tile size $T$ as being set first, before the network architecture is fixed, then the choice of networks is restricted to architectures that produce $H/T \times{} W/T$ pixel outputs. Alternatively, we may take the perspective that the choice of network architecture determines tile size $T$.
% The output of the image backbone has spatial dimensions $H' \times W'$.

In addition to determining (or respecting) the tile size, the image backbone architecture, especially its depth, kernel sizes, and strides, determines the receptive field for each tile; that is, which pixels in the image can influence the distributional parameters for each tile's latent variables.
We use the nonnegative integer $r_{\mathcal X}$ to characterize the ``radius'' of the receptive field of the image backbone.
In particular, for tile $\ell \in \Omega$, the output of $f_{\mathcal X}$ in spatial position $\ell$ only depends on pixels within $r_{\mathcal N}$ pixels of pixel $x_{(\ell - 1)T+1}$, in terms of the $L_\infty$ norm.

\paragraph*{The neighborhood network}
As input, the neighborhood network $f_{\mathcal N}$ takes 1) the properties of objects in tiles with rank $k$ that are before $i$ in the ordering induced by $\sigma$, and 2) the properties of objects in tiles with pre-$k$ ranks. 

The network produces output in two steps. In the first step, regardless of $i$ or $k$, it constructs a representation of these objects with dimension $D_1 \times H' \times W'$, where $D_1$ is a user-defined number of channels.
This representation is constructed without trainable parameters, by concatenating objects' properties in a way that preserves their spatial positions, zero-padding as necessary. 
%It is not necessary to include all the objects' properties, but we always include one channel with the number of objects per tile. Richer conditioning information makes the variational distribution more expressive, but can hamper generalization.

In the second step, this initial representation is transformed by a fully convolutional network to a tensor of dimension $D_2 \times H' \times W'$, where $D_2$ is a user-defined embedding dimension.
This network is lightweight relative to the image backbone; it is composed of just a few convolutional blocks.

We use the nonnegative integer $r_{\mathcal N}$ to characterize the ``radius'' of the receptive field of the neighborhood network, using the $L_\infty$ norm. In particular, for $\ell \in \Omega$, the output of $f_{\mathcal N}$ in spatial position $\ell$ only depends on $z_{\ell'}$ if $\|\ell - \ell'\|_\infty \le r_{\mathcal N}$.
For example, if the neighborhood network contains only 1$\times$1 kernels except for a single 3$\times$3 kernel (and stride is always 1), then $r_{\mathcal N}=1$.

%The within-tile context network contains only convolutions with 1$\times$1 kernels, which prevents it from encoding spatial variation.
%It produces a representation of the objects so far in each tile. Whereas the neighborhood context network is run $K$ times to sample the variational distribution (or to evaluate the log-likelihood), the within-tile context network is run $M$ times for each of $K$ ranks.

\paragraph*{The detection head}
The detection head takes input from both the image backbone and the neighborhood network. Using five convolutional blocks exclusively with 1$\times$1 convolutional kernels, it outputs the distributional parameters for each object inferred. The precise number of these parameters depends on the variational family; it is on the order of tens in our case studies.

\subsection{Stuctural equivalence between the posterior and the variational distribution}
\label{sec:mirrors}
In the generative model, the latent variables for different tiles are sampled independently, and pixel intensities depend solely on nearby tile latent variables (\cref{sec:model}).
In contrast, the variational distribution is defined by an autoregressive process: for each $k$, the rank-$k$ tile latent variables depend only on pre-rank-$k$ tile latent variables. Further, the tile latent variables for a particular tile $\ell$ depend only on the tile latent variables of tiles near $\ell$ (\cref{sec:vf_form,sec:amortized}).

Both the generative model and the proposed family of variational distributions make conditional independence assertions through their factorizations.
Ideally, these assertions are the same \citep{webb2018faithful}.
Independence assertions made by the posterior but not the variational family are problematic because costly and inexact numerical approximation would be required to infer what is known a priori.
Independence assertions made by the variational family but not the posterior limit are potentially even more problematic: no amount of optimization would allow the recovery of these dependencies.

The proposed variational family can be configured to match the independence assertions of the posterior distribution, in the sense of being a minimal I-map for the posterior distribution, as we show in Appendix~B.
This configuration involves a particular setting of the number of ranks $K$, the image receptive field radius $r_{\mathcal X}$, and the catalog receptive field radius $r_{\mathcal N}$, which all involve the maximum object radius $R$. To exactly match the structure of the posterior, $K$ needs to grow quadratically in the ratio between object radius $R$ and tile size $T$. Setting $K$ to a smaller value would sever some dependencies present in the exact posterior distribution, but may nevertheless be preferable if doing so reduces the computational burden.

% In this variational family, the tile variables with the same rank are conditionally independent. To sample from the variational distribution, the latent variables of tiles with the same rank can therefore be sampled in parallel. In the special case that the number of ranks $K$ is one, our variational distribution is equivalent to the independent tiling of \citet{liu2023variational}, and all tiles could be sampled in parallel. With $K > 1$ but less than the number of tile indices $|\Omega|$, some tiles can be sampled in parallel, but the dependence structure is richer than with $K = 1$.

% Due to the spatial structure of images and the limited long-range dependence among random variables in the posterior, the number of ranks $K$ can often be much smaller than the number of tile indices $|\Omega|$, while still giving $q$ the flexibility to model all conditional dependencies in the posterior. (We prove this in \cref{sec:blanket}.)

\section{Fitting the variational distribution}
\label{sec:fitting}

Let $\phi$ denote the parameters of the variational distribution.
To make explicit the dependence of $q$ on $\phi$, set $q_\phi \coloneqq q$.
Traditionally, in variational inference, the variational distribution $q_\phi$ is fitted to minimize the following Kullback-Leibler (KL) divergence:
\begin{align*}
    \mathcal{L}_{\text{rev}}(\phi) \coloneqq D_\mathrm{KL}(q_\phi(z \mid x), p(z \mid x)).
\end{align*}
KL divergence is not symmetric. This one, which is from the approximation to the target, is known as the reverse (or exclusive) KL divergence. In contrast, the KL divergence from the target to the approximation is known as the forward (or inclusive) KL divergence \citep{zhang2023transport}.
The reverse KL divergence is equivalent to the negative ELBO up to a constant with respect to $\phi$~\citep{blei2017variational}. Thus, maximizing the ELBO to fit $q_\phi$ also minimizes the reverse KL divergence.

Instead of using this traditional variational objective, we propose to fit $\phi$ to minimize an expectation of the forward KL divergence:
\begin{align*}
    \mathcal L_\mathrm{fwd}(\phi) \coloneqq \mathbb{E}_{x \sim p(x)} D_{\mathrm{KL}}(p(z \mid x) \mid\mid q_\phi(z \mid x))
%  &= \mathbb{E}_{x \sim p(x)} \mathbb{E}_{z \sim p(z \mid x)} \left[ \log \frac{p(z \mid x)}{q_\phi(z \mid x)}\right].
\end{align*}
This expectation is taken with respect to samples from the model; unlike in ELBO-based VI, it is not an empirical average of observed data points.
In a sense, fitting this objective function requires us to ``solve the inverse problem'' for all possible synthetic data sets, rather than just the observed data.

Remarkably, 
the gradient of this objective can be efficiently estimated without bias, which allows us to minimize the objective using stochastic first-order optimizers such as stochastic gradient descent (SGD). 
Observe that due to Bayes' rule,
\begin{align}
\mathcal L_\mathrm{fwd}(\phi)
&= -\mathbb{E}_{x \sim p(x)} \mathbb{E}_{z \sim p(z \mid x)} \left[\log q_\phi(z \mid x) \right] + c\\
&= -\mathbb{E}_{z \sim p(z)} \mathbb{E}_{x \sim p(x \mid z)} \left[\log q_\phi(z \mid x) \right] + c,
\label{eq:favi_ancestral}
\end{align}
where $c$ is constant with respect to $\phi$.
Thus, stochastic approximations of $\mathcal L_\mathrm{fwd}(\phi)$ can be formed through ancestral sampling: sample $z$ from the prior and sample $x$ conditional on $z$.
Computing this approximation using automatic differentiation gives an unbiased approximation of the gradient:
\begin{align}
\nabla \mathcal L_\mathrm{fwd}(\phi)
&= -\mathbb{E}_{z \sim p(z)} \mathbb{E}_{x \sim p(x \mid z)} \left[\nabla \, \log q_\phi(z \mid x) \right].
\label{eq:grad_favi}
\end{align}

Once $\phi$ has been fitted,  the posterior is approximated by $q_{\phi^\star}(z \mid x_0)$, where $x_0$ is the observed data and $\phi^\star$ is the optimizer of \cref{eq:grad_favi}. This approach to inference, now known as NPE, has evolved through a series of works \citep{barber2004im,bornschein2015reweighted,papamakarios2016fast,ambrogioni2019forward}.

\subsection{Comparison of ELBO-based VI and NPE}
NPE has advantages over traditional ELBO-based VI, including that 1) unbiased low-variance stochastic gradients of its objective can be easily computed, even with a mix of discrete and continuous random variables (the reparametrization trick is not needed); 2) inference does not require computing likelihoods; 3) latent variables can be marginalized over implicitly, simply by excluding them from the optimization problem; 4) if the variational family is not flexible enough to include the exact posterior, the posterior approximation recovered is generally overdispersed rather than underdispersed, which is typically preferable; and 5) with plausible assumptions, optimization in NPE is globally convergent regardless of the contours of the log likelihood function. Appendix~C provides further details about these advantages, as well as several drawbacks of NPE.

\subsection{Efficient NPE for fitting the proposed variational distribution}
\label{sec:our-npe}

The NPE objective (\cref{eq:favi_ancestral}) can be further expanded by substituting in our particular variational distribution (\cref{eq:q_ranks}), yielding
\begin{align}
\mathcal L_\mathrm{fwd}(\phi)
&= -\mathbb{E}_{z \sim p(z)} \mathbb{E}_{x \sim p(x \mid z)} \left[\sum_{k = 1}^K \log q_{\phi}(z_k \mid z_{<k}, x) \right] + c.
\label{eq:tiles}
\end{align}

Minimizing this objective can be viewed as learning to solve a class of supervised infilling problems: given some entries of a catalog (specifically, the pre-rank-$k$ tile latent variables $z_{<k}$), predict other catalog entries (the rank-$k$ tile latent variables $z_k$). The image $x$ is also provided context, but it is the catalog---not the image---that we are infilling. 

Each optimization step proceeds as follows: (1) load a simulated image $x$ and its corresponding catalog $z$ to the GPU; (2) create a masked catalog by zeroing out all entries of $z$ except $z_{<k}$, which serve as predictors for $z_k$; (3) pass the image and the masked catalog through the inference network, which outputs predictions for all catalog entries; (4) compute the loss (negative log likelihood) for the predicted values, but apply masking to retain loss only for the target entries $z_k$; and (5) update the network weights using gradients computed via automatic differentiation.

\subsection{Network hyperparameter selection}
Neural networks require optimizing hyperparameters across vastly higher-dimensional spaces than traditional statistical methods. These include choices such as the optimizer type and its configuration, normalization strategy, architectural parameters (for example, the number of units per layer, kernel dimensions, and padding), learning rate schedule, activation function, and structural modules like residual blocks or attention heads.

The standard approach to hyperparameter selection involves derivative-free optimization: evaluate multiple configurations and choose the one that yields the lowest validation loss. Since hyperparameters are not directly of scientific interest, it is sufficient to identify workable values rather than to find the global optimum; even limited exploration generally yields models that outperform traditional approaches.

\citet{karpathy2019recipe} recommends starting with a simple baseline model---perhaps copied from related literature---and then adding complexity one change at a time while validating each modification's impact on validation loss. This incremental approach ensures that each added architectural or tuning decision contributes demonstrably to performance.
\citet{bergstra2012random} demonstrate both empirically and theoretically that, when given the same computational budget, random search outperforms grid search, particularly if only some hyperparameters significantly influence validation performance.

\subsection{Setting the number of ranks}
In principle, the number of ranks $K$ could also be included in a hyperparameter search targeting validation loss. However, fixing $K = 4$ is adequate to ensure no two neighboring tiles (adjacent or diagonally adjacent) have the same rank (cf. Figure 1). As long as neighboring tiles have different ranks, they will not be cataloged independently, which is adequate to solve the problem that motivated this work: that objects on tile boundaries are often detected twice or missed if neighboring tiles are treated as independent in the variational distribution. Hence, we recommend fixing $K = 4$ for cataloging problems and follow this advice in our subsequent case studies. For other problems, which require modeling long-range posterior dependencies, it may be beneficial to explore larger $K$.

\subsection{Setting the tile size and the maximum number of objects per tile}
The choices of tile size $T$ and maximum sources per tile $M$ are interrelated. We recommend setting $T$ and $M$ so that the overwhelming majority of tiles are expected to contain no more than $M$ objects of interest, according to Poisson statistics. This can be achieved by making $T$ small, $M$ large, or some combination thereof. Any choice of $T$ and $M$ entails trade-offs. Large $M$ requires more sequential computation, as $M$ forward passes are needed, and the loss function requires $M!$ computations. These $M!$ calculations are individually inexpensive, and thus of little consequence for $M < 7$, but prohibitive for $M > 10$.
Also, larger $M$ results in a deeper hierarchy in the variational distribution, which could potentially exacerbate the issues of ``exposure bias,'' discussed in \cref{sec:npe-discussion}.

Small $T$, on the other hand, corresponds to more tiles for a fixed-size image. Tiles can be processed in parallel, so processing many tiles does not require more sequential compute, but it would require more GPU memory, and would likely become a bottleneck for $T < 1$ because the output size would exceed the input size.
A potential benefit of small $T$ is that there is less diversity of patterns in smaller tiles, which could lead to faster training. However, shrinking the tiles likely requires the network to rely more on the neighborhood network for context.

Although there are many competing concerns in selecting $M$ and $T$, as with setting neural network hyperparameters more generally, we should not be paralyzed by the lack of definitive guidance on finding the best single value: many settings can work well and the values that have already been demonstrated to work well are a good place to start.
The various metrics introduced in the case studies can give us quantitative guidance on our choices of $M$ and $T$.

%%%%%%%%%%%%%%%%%%%%%%%%%%%%%%%%%%%%%%%%%%%%%%%%%%%%%%
\section{Case studies} \label{sec:experiments}
%%%%%%%%%%%%%%%%%%%%%%%%%%%%%%%%%%%%%%%%%%%%%%%%%%%%%%

%The expressivity of the proposed variational family gives us some reason to expect good performance from the proposed method. However, it is largely an empirical question as to 1) whether today's neural network architectures and training protocols can deliver adequate performance for this scientific application, 2) whether the autoregressive structure is necessary and sufficient for approximating the posterior with enough fidelity to be scientifically relevant, and 3) whether the generative model is accurate enough that inferences derived from it provide useful insights.

We conduct two case studies, each culminating in the analysis of real astronomical images from the Sloan Digital Sky Survey (SDSS), and each with a comparison of the proposed method to existing methods.
In the first case study, we catalog a high-resolution SDSS image that contains a typical density of detectable stars and galaxies ($\sim$5 objects per square arcminute).
In the second case study, we catalog an SDSS image of a crowded starfield known as the Messier 2 (M2) globular cluster, with a much higher density of detectable objects: $\sim$2500 objects per square arcminute.
For comparison, the Rubin Observatory LSST pipeline will detect $\sim$500 objects per square arcminute \citep{RubinObservatoryKeyNumbers}.
% Taken together, the object densities in these two case studies upper and lower bound the object densities that we are likely to encounter in images from future surveys.

%Leading up to the analysis of a real SDSS image, in each case study, we generate synthetic data like the targeted SDSS image, train an inference network (i.e., fit a variational distribution) using a subset of this synthetic data, and use held-out synthetic data, for which we know ground truth, to evaluate performance.

In addition to comparing to published methods, we evaluate two variants of the proposed method.
The first variant has the number of ranks $K$ set to one, which implies that the latent variables corresponding to different tiles are independent in the variational distribution. This setting of $K$ provides a baseline without spatially autoregressive structure. 
The second variant has the number of ranks $K$ set to four, so that no adjacent tiles have the same rank.
We do not consider $K>4$ because four ranks should be adequate to resolve the concern that motivated this work: that it can be ambiguous which neighboring tile contains an object's center.

In both case studies, the proposed method is implemented within the software framework of the Bayesian Light Source Separator (BLISS), and we refer to the proposed method as BLISS throughout this section. The BLISS codebase began in 2020 as a fork of StarNet~\citep{liu2023variational}. Since then, it has been adapted to catalog galaxies as well as stars~\citep{hansen2022scalable}, to detect and characterize gravitational lenses~\citep{patel2022scalable}, and to support spatially varying point PSFs and backgrounds~\citep{patel2025neural}. All of these previous works have relied on independent tiling.

\subsection{Case study 1: A typical SDSS image}

Our first case study focuses on cataloging a typical SDSS image, which contains roughly one object per 4000 pixels that is brighter than 22.5 magnitude, which is sometimes taken to be the SDSS detection threshold.
Note that while both ``magnitude'' and ``flux'' quantify the brightness of objects, magnitude is an inverse logarithmic transformation of flux; brighter objects have lower magnitudes.
% Pixel scale is 0.396 arcseconds.

\subsubsection{Data generation}

The first step in applying BLISS to SDSS images is to generate synthetic SDSS-like data (that is, catalogs and images) through ancestral sampling for use in training the BLISS inference network.
Our simulated data contains 52\% galaxies and 48\% stars, which is in line with existing SDSS catalogs.
Our prior is hierarchical, and the object flux distributions differ for stars and galaxies.
Both follow truncated Pareto distributions with parameters fitted to existing catalogs.
Objects are distributed uniformly in images.
Galaxies, which have spatial extent, follow a six-parameter bulge-and-disk model, as in \cite{regier2019approximate}.

We draw 131,072 object catalogs from this prior, and render a 256$\times$256-pixel image for each using the GalSim image simulator~\citep{rowe2015galsim}, which implicitly defines the conditional likelihood of our generative model. 
In addition to taking a catalog as an argument, GalSim requires us to specify 1) a point spread function (PSF), 2) calibration constants, which map counts of photoelectrons that a camera records to physical units, and 3) background intensities.
We set these three arguments to GalSim based on the estimates of them from the SDSS photometric pipeline.

% Generation takes 1.5 hours using 16 processes; runtime is dominated by calls to GalSim. 

\subsubsection{Fitting the variational distribution}
We configure the BLISS inference network to detect at most $M=1$ object per 4$\times$4-pixel tile, which is adequate to recover 99.6\% of objects according to Poisson statistics.
The variational distribution for each tile is governed by 80 distributional parameters, which are outputted by the inference network. These distributional parameters govern the object type (Bernoulli, 1 parameter), the flux (Gaussian, 2 parameters), the galaxy shape (composite, 12 parameters), and the within-tile position (a compound categorical-uniform distribution, 64 parameters).
We use the Adam optimizer with a constant learning rate.  
Training takes 5 hours on an NVIDIA 2080 Ti GPU with a batch size of six.
We use data augmentation, rotating by multiples of 90 degrees and flipping each image and the corresponding catalog.

%%%%%%%%%%%%%%%%%%%%%%%%%%%%%%%%%%%%%%%%%%%%%%%%%%%%%%
\subsubsection{Probing the response curve}
\label{sec:toy-data}
%%%%%%%%%%%%%%%%%%%%%%%%%%%%%%%%%%%%%%%%%%%%%%%%%%%%%%

Ideally, an inference network would detect objects equally well in the interior of an image regardless of their positions with respect to tile boundaries.
To assess how nearly various inference networks achieve this ideal, we apply them to simulated images of a star at various positions that are all near a tile boundary (\cref{fig:moving_star}).
% We expect that an inference network with independent tiling (i.e., $K=1$) will perform poorly in this setting: if it is ambiguous which tile contains the star, then multiple tiles have a non-negligible probability of containing the star under the posterior. Therefore, there should be strong negative dependence in the posterior: if one tile contains the star, the other cannot (at least not the same star). With $K=1$, the variational dependence cannot encode this negative dependence, whereas with $K > 1$, the structure of the variational distribution allows it.

\begin{figure}[t]
    \centering
    \includegraphics[width=.25\textwidth]{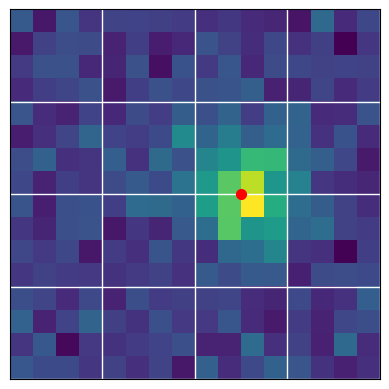}
    \includegraphics[width=.25\textwidth]{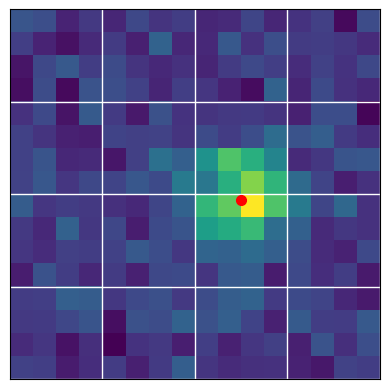}
    \includegraphics[width=.25\textwidth]{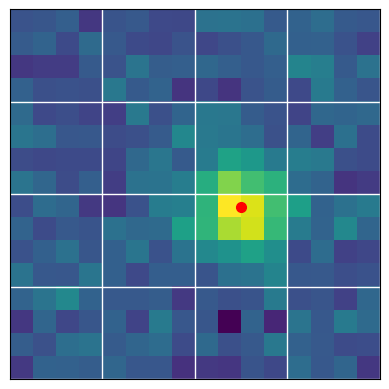}
    \caption{Three simulated images of a star. The tiles, which are $4\times{}4$ pixels in size, are demarked by thin white lines. The red dot in each image indicates the star's position within that image. In the left image, the star is imaged at a tile boundary. In the center image, the star is imaged 0.27 pixels below the tile boundary. In the right image, the star is imaged 0.54 pixels below the tile boundary.
    }
    \label{fig:moving_star}
\end{figure}

\paragraph*{One star directly on a tile border}
First, we apply BLISS with autoregressive tiling to images of a star precisely positioned on the border between two tiles (\cref{fig:border-experiment}).
The marginal detection probabilities show that a non-negligible probability is assigned to both of the tiles that border the object. These probabilities sum to one in this instance, which is not the case by construction, so this is encouraging: with enough training data, BLISS should learn this relation.
The remaining tiles all have a negligible probability of containing an object. %Because objects of interest are typically bright enough that they can be distinguished from the background, the marginal probability of a detection is low.

Conditional on the event that the object is not contained in the lower of the two tiles it borders, the probability of the object being contained in the upper tile is 100\% (center-right). Counterfactually, had no objects been detected in the lower tile, the upper tiles would have a 4\% probability of containing an object (right). This is reasonable for our data as a bright neighboring star can effectively ``mask'' or ``hide'' a faint star in a neighboring tile.

\begin{figure}
    \centering
    \includegraphics[width=0.235\textwidth]{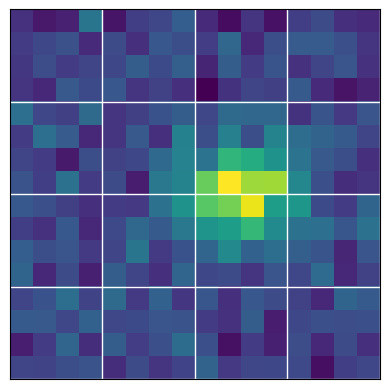}
    \includegraphics[width=0.24\textwidth]{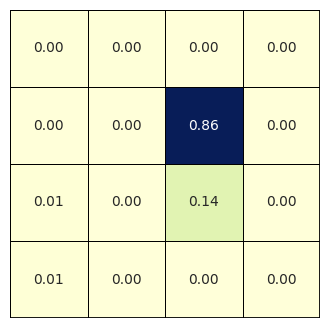}
    \includegraphics[width=0.24\textwidth]{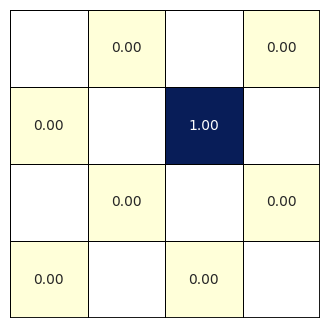}
    \includegraphics[width=0.24\textwidth]{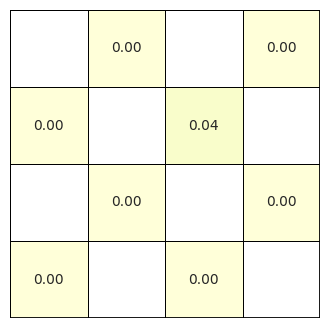}
    \caption{Analyzing a simulated image of a star centered on the border of two tiles. Left) the image. Center-left) the marginal probability of detection in each tile under the BLISS variational distribution. Center-right) the conditional probabilty of detection given no detections in the white tiles. Right) the conditional probability of detection given a single object was detected in tile (3, 3).}
    \label{fig:border-experiment}
\end{figure}

\paragraph*{One star ``moving'' across tile border boundaries}

Whereas the former experiment considered detection probabilities for a particular image,
we now average over 100 images for each catalog, each created by adding Poisson shot noise to a noise-free rendering of the catalog.
We then examine the frequency of detecting the correct number of objects.
We consider a variety of catalogs, each containing one star shifted vertically by a variable amount.
\cref{fig:correct_count_prob}, which is a key result of this article, shows the detection rates for both independent tiling and autoregressive tiling.

For a bright star (20.47 magnitude), with autoregressive tiling, the probability of detecting the correct number of objects (one) is 100\% regardless of position.
With independent tiling, accuracy decreases from 100\% in the center of a tile to 65\% at the border.
This decrease is highly undesirable because it is an artifact of tile boundaries, which were artificially introduced by our inference procedure.

For a faint star (22.21 magnitude), with autoregressive tiling, the frequency of detecting the correct number of stars (that is, one) was around 95\% regardless of position. In contrast, with independent tiling, the accuracy decreased from 95\% when the star was in the center of the tiles to 60\% at the borders of the tiles, showing a strong dependence on position.

\begin{figure}[t]
    \centering
    \includegraphics[width=.48\textwidth]{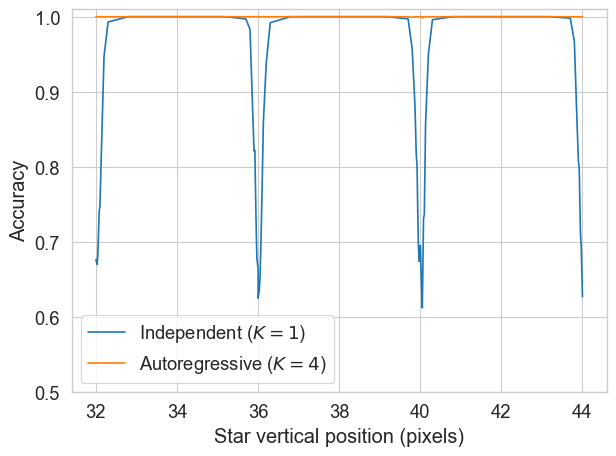}
    \includegraphics[width=.48\textwidth]{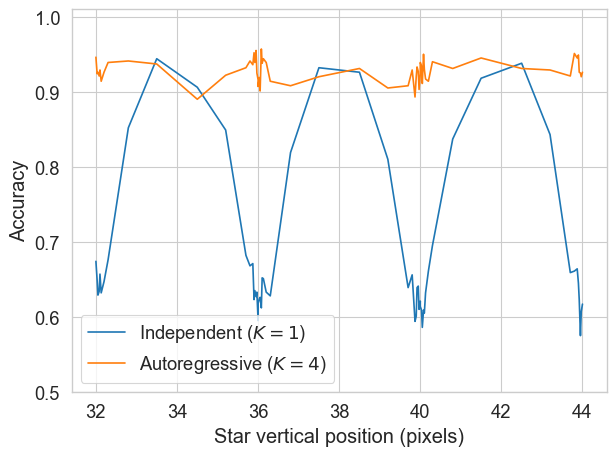}
    \caption{The probability under the variational distribution of correctly detecting exactly one star, for simulated images of one star in a variety of positions. The simulated stars are either bright with a magnitude of 20.47 (left) or faint with a magnitude of 22.21 (right).
    Tile boundaries are at positions that are multiples of four.
}
    \label{fig:correct_count_prob}
\end{figure}

\paragraph*{A blended star and galaxy}

Now we repeat parts of the previous experiment in the presence of blending. In particular, we consider a synthetic image of both a star on a tile border and a galaxy several pixels away (\cref{fig:with-blending}).
In this case, under the variational distribution, the marginal probability of detecting the galaxy, which is near the center of a tile, is 100\%, whereas the marginal probability of detecting the star, which is on a border, is 33\% in the upper tile and 75\% in the lower one.

If we condition on the event of no detections in the white squares of our checkerboard, which includes the tile below the star, then the probability of an object in the upper tile increases to 100\%, as desired.
Conversely, given the event that the white squares included a detection in the tile below the star, then it becomes unlikely that the upper tile contains an object.

\begin{figure}
    \centering
    \includegraphics[width=0.235\textwidth]{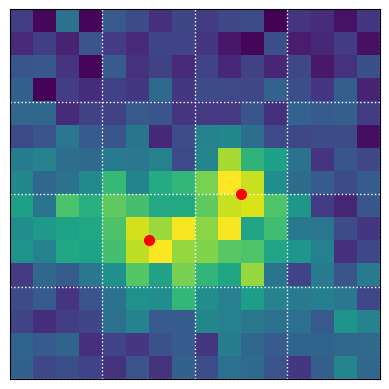}
    \includegraphics[width=0.24\textwidth]{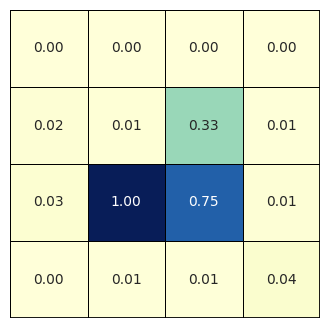}
    \includegraphics[width=0.24\textwidth]{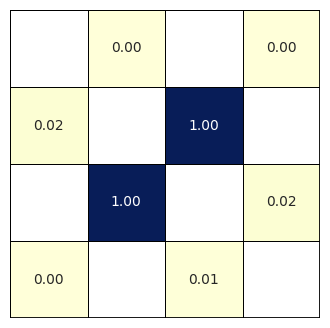}
    \includegraphics[width=0.24\textwidth]{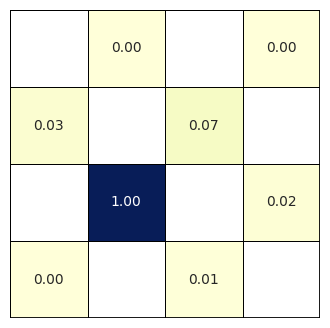}
    \caption{Analyzing a simulated image of a blended star on a tile border and a galaxy. Left) the image. Center-left) the marginal probability of detection in each tile under the BLISS variational distribution. Center-right) the conditional probability of detection given no detections in the white tiles. Right) the conditional probability of detection given a single object was detected in the tile (3, 3).}
    \label{fig:with-blending}
\end{figure}

% For the synthetic catalogs used to generate \cref{fig:with-blending}, over many draws of images with variable Poisson noise per pixel, we can compute the likelihood of the catalog.
% With independent tiling, the log-likelihood of the ground-truth catalog is -202.55 on average, whereas with spatially autoregressive tiling the average log-likelihood is higher: -44.92.
% This difference in log-likelihood is highly repeatable over draws of the images and has an intuitive basis: autoregressive tiling, based on a variational distribution with $K=2$, makes inference about half of the random variables with knowledge of the other half, whereas independent tiling does not.

%%%%%%%%%%%%%%%%%%%%%%%%%%%%%%%%%%%%%%%%%%%%%%%%%%%%%%
\subsubsection{Simulations for assessing aggregate predictive performance}
\label{sec:sdss_field_synth_prediction}
%%%%%%%%%%%%%%%%%%%%%%%%%%%%%%%%%%%%%%%%%%%%%%%%%%%%%%

We now assess the quality of our fitted variational distribution by considering its aggregate performance on a large held-out simulated data set.

\cref{tab:sdss_synth_performance} shows that the log-likelihood of the held-out ground truth catalogs under the variational posterior increases as the number of ranks $K$ increases.
The increase from increasing $K$ from one (independent tiling) to four ranks is highly repeatable. We can understand this increase intuitively through the lens of predictive performance: with more ranks, inferences for some tiles are based on more information, and no inferences are based on less information.

Although increases in log-likelihood are encouraging, establishing the practical importance of these increases requires additional probes.
It is possible, due to the limitations of forward KL minimization (cf. Appendix~C), that increasing the log-likelihood $q(z \mid x)$ could make our posterior approximations less useful.

\begin{table}
\centering
\begin{tabular}{c|c|c|c|c}
           & log-likelihood & precision & recall & F1 score \\
           \hline
Independent ($K=1$) & -227,337 (2284)  & 0.9734 (0.0005) & 0.8556 (0.0010)  & 0.9101 (0.0006) \\
Autoregressive ($K=4$) & -210,833 (2181)  & 0.9633 (0.0006) & 0.0.8861 (0.0009) & 0.9227 (0.0006) \\
\end{tabular}
\caption{BLISS performance with synthetic data that mimics typical SDSS images. The standard error appears in parentheses.}
\label{tab:sdss_synth_performance}
\end{table}

We first consider the precision and recall of the mode of various variational distributions.
The definitions of precision and recall that we use are generalizations of the standard notions of precision and recall, which make these concepts applicable to spatial data.
Specifically, we form a maximal bipartite matching between the detections (according to the mode of the variational distribution) and the objects (ground truth), with the restriction that a detection and an object can match only if they are within 2 pixels of each other.
Recall is the proportion of objects matched, and precision is the proportion of detections matched.

\cref{tab:sdss_synth_performance} shows that increasing the number of ranks $K$ increases the F1 score (that is, the harmonic mean of precision and recall).
As $K$ increases from one to four, recall improves by 3\%, which is a  21\% reduction in the number of missed objects. This increase in recall is accompanied by a smaller (1\%) decrease in precision. 

%%%%%%%%%%%%%%%%%%%%%%%%%%%%%%%%%%%%%%%%%%%%%%%%%%%%%%
\subsubsection{Simulation-based calibration}
\label{sec:sdss_field_synth}
%%%%%%%%%%%%%%%%%%%%%%%%%%%%%%%%%%%%%%%%%%%%%%%%%%%%%%

Although interpretable, precision and recall are not comprehensive metrics for assessing the quality of posterior approximations, as they are based on a single point estimate. The point estimate we use above, namely, the mode, assigns great importance to the 50\% threshold for detection, for example.

Now, we assess the quality of the variational distribution by comparing samples from it to samples from the posterior distribution.
We can readily sample the variational distribution $q(z \mid x)$ for arbitrary $x$, but we cannot do so for the exact posterior.
However, we can implicitly sample $p(z \mid x)$ once for each $x$ sampled from $p(x)$, because, by Bayes' rule, $p(z \mid x)p(x) = p(x \mid z)p(z)$; the latter distribution can be sampled through ancestral sampling. Thus, for arbitrary subsets of latent space $\mathcal I$ and $\mathcal J$, we can use Monte Carlo sampling to estimate
\begin{align*}
C(\mathcal I, \mathcal J) \coloneqq
    \mathbb E_{x \sim p(x), z' \sim p(z \mid x), z'' \sim q(z \mid x)} \left[  \mathbf 1\{z' \in \mathcal I \wedge z'' \in \mathcal J \} \right].
\end{align*}
If $q(z \mid x)p(x)$ is equal in distribution to $p(z \mid x)p(x)$, then $C(\mathcal I, \mathcal J) = C(\mathcal J, \mathcal I)$ for arbitrary $\mathcal I$ and $\mathcal J$.
Further, if $\mathcal I$ and $\mathcal J$ are themselves stochastic, then 
$\mathbb E_{\mathcal I, \mathcal J}C(\mathcal I, \mathcal J) = E_{\mathcal I, \mathcal J} C(\mathcal J, \mathcal I)$.
The implication of this result for us is that the quality of various posterior approximations can be assessed in terms of the extent to which they preserve this symmetry. This metric is based on the same principles as the test of variational simulation-based calibration (VSBC) proposed by \cite{yao2018yes} and further developed by \cite{modrak2025simulation}.

\begin{figure}[t]
    \centering
    \includegraphics[width=0.4\linewidth]{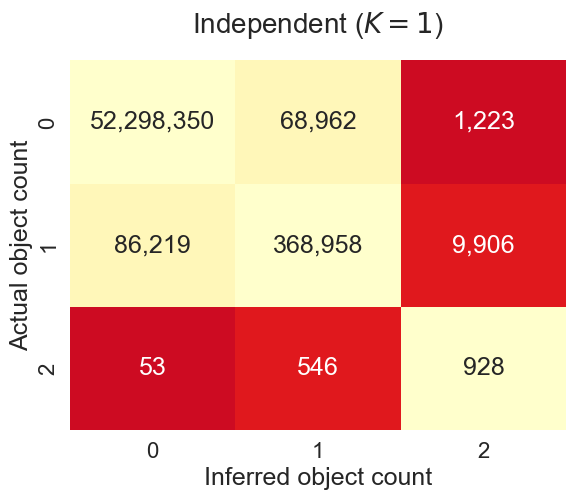}
    \hspace{2em}
    \includegraphics[width=0.4\linewidth]{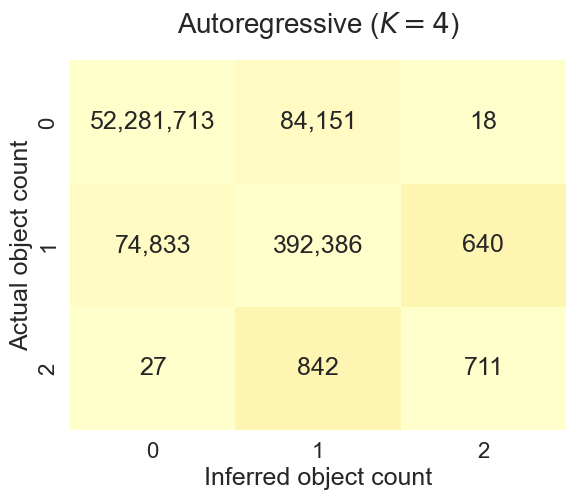}
    \caption{Confusion matrices for samples of catalogs from the BLISS variational distribution for synthetic images that mimic typical SDSS images. For each catalog, whether inferred or actual, source count is the number of objects in various $8\times$8-pixel regions.
    Cell colors indicate the difference between each transpose pair's entries as a proportion of the smaller of the two entries in the pair.
    }
    \label{fig:confusion_matrices}
\end{figure}

We take $\mathcal I$ and $\mathcal J$ to each be the event that the sampled catalog $z$ contains a particular number of astronomical objects in an $8\times{}8$-pixel region covered by four adjacent tiles. The particular block of adjacent tiles is selected uniformly at random, meaning that $\mathcal I$ and $\mathcal J$ are themselves stochastic.
The estimates of $E_{\mathcal I, \mathcal J} C(\mathcal J, \mathcal I)$ for the possible source counts can be presented as a confusion matrix, with the catalog implicitly sampled from $p(z \mid x)$ as the actual catalog used to generate the image, and the catalog sampled from $q(z \mid x)$ as the prediction. 
Typically in confusion matrices, good performance equates with high values along the diagonal. However, because our ``predictions'' are samples rather than, say, the mode of the variational distribution, we do not necessarily expect high values along the diagonal. Instead, we look for symmetry in ``transpose pairs'' of off-diagonal entries; that is, pairs of off-diagonal matrix entries with the property that if one has the coordinates $(i,j)$ the other has the coordinates $(j, i)$.

Figure~\ref{fig:confusion_matrices} provides these confusion matrices for variational distributions with $K=1$, 2, and 4, for a large held-out synthetic data set.
We are particularly interested in the transpose pair that describes how frequently one object is mistaken for two and vice versa.
The variational distribution with $K=1$ is 18 times more likely to predict that one object is two than to predict that two objects are one.
This extreme imbalance is largely corrected by the autoregressive variational distribution, which infers 640 times that one object is two and 842 times that two objects are one.

%%%%%%%%%%%%%%%%%%%%%%%%%%%%%%%%%%%%%%%%%%%%%%%%%%%%%%
\subsubsection{Application to a two-megapixel SDSS image}
\label{sec:sdss_field}
%%%%%%%%%%%%%%%%%%%%%%%%%%%%%%%%%%%%%%%%%%%%%%%%%%%%%%
We apply BLISS to an image produced by the Sloan Digital Sky Survey (SDSS) that is two megapixels in size (1488$\times$1488 pixels), covering 96 arcmin$^2$.
Although no ground-truth catalogs are available, we can use a catalog produced by the Dark Energy Camera Legacy Survey (DECaLS) as a proxy for ground truth.
This DECaLS catalog is based on deeper and higher-resolution images of the same region.

Note, however, that the DECaLS catalog is not a perfect proxy for ground truth for several reasons. First, because the DECaLS catalog is based on more recent images than SDSS, the positions of some nearby stars may have shifted multiple pixels during the time between surveys. Second, the DECaLS cataloging software makes errors of its own. Third, we consider just a subset of the DECaLS catalog containing hundreds of detectable objects, which is not enough to get stable estimates for percent-level changes in accuracy. For all of these reasons, this real data analysis should not be taken as a substitute for our early results with synthetic data, but rather should be viewed as complementary.

Approximate posterior inference requires 0.3 seconds for the full two-megapixel image; the image can be processed in a single forward pass. By requiring just a single forward pass of the network, inference requires little runtime on top of the time required for data loading, implying that the method is scalable to even the largest data sets.

% Recall that our inference network was trained with 256$\times$256-pixel images. Because the network is fully convolutional, we can either apply the network directly to the 1488$\times$1488 image, or partition the image into (potentially overlapping) 256$\times$256 subimages and apply the network to a batch of these subimages. Performance is much better with the latter approach. Although fully convolutional, the network is sensitive to the image size because the network is trained in the context of a certain amount of zero padding.

\begin{figure}
    \centering
    \includegraphics[width=1\linewidth]{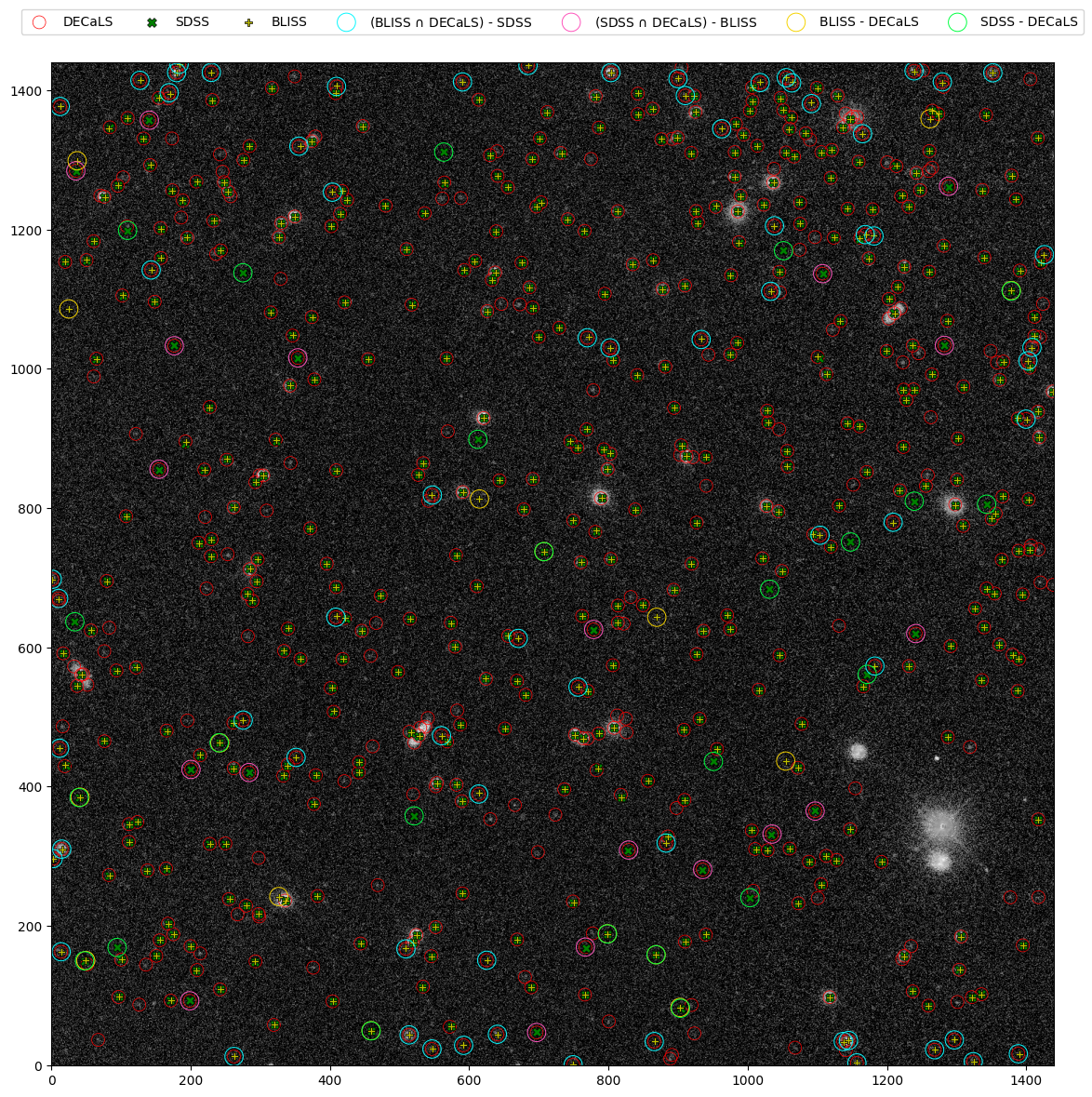}
    \caption{The studied SDSS image overlayed with detections by our software (BLISS), the image processing pipeline of the Sloan Digital Sky Survey (Photo), and the software of the Dark Energy Camera Legacy Survey (DECaLS).}
    \label{fig:sdss_four_colors}
\end{figure}

\Cref{fig:sdss_four_colors} shows the detections for BLISS (with four ranks), DECaLS (our proxy for ground truth), and the SDSS catalog (a competitor). The SDSS catalog was created by software called Photo~\citep{lupton2001sdss}. In contrast to DECaLS, Photo was applied to the same data as BLISS. We consider just objects brighter than 22.5 magnitude, which is approximately the detection threshold.
We omitted two regions of the image containing saturated pixels due to the presence of bright stars, as our simulator does not yet account for saturated pixels.

\begin{table}[t]
    \centering
    \begin{tabular}{l|ll}
                       & true positives & false positives \\
    \hline
    Photo~\citep{lupton2001sdss}   & 431  & 11  \\
    BLISS, $K=1$   & 488  &  28 \\
    BLISS, $K=4$ & 511  &  30 \\
    BLISS, $K=4$, conservative & 456 & 11 \\
    \end{tabular}
    \caption{BLISS detection performance for a particular SDSS image (run 94, field 1, camcol 12). The ``conservative'' version of BLISS has a lower detection threshold, allowing it to strictly outperform Photo.} 
    \label{tab:sdss_field_scores}
\end{table}

\cref{tab:sdss_field_scores} summarizes the detection performance of BLISS and Photo. Regardless of the number of ranks used, BLISS correctly detects at least 57 more objects than Photo while making at most 29 more false detections. By lowering the detection threshold, we created a ``conservative'' version of BLISS that strictly outperforms Photo: it has the same precision as Photo but correctly discovers 25 additional objects. 

For BLISS, the $K=4$ settings outperform the $K=1$ setting by roughly the amount expected based on simulations.
However, the overall recall of BLISS is 10\% lower than projected by simulation, which points to some degree of model misspecification. This degree of misspecification is likely acceptable given the superior performance of BLISS here, but performance could potentially be improved by enhancing the realism of the simulator.

%%%%%%%%%%%%%%%%%%%%%%%%%%%%%%%%%%%%%%%%%%%%%%%%%%%%%%
\subsection{Case study 2: A crowded starfield} \label{sec:m2}
%%%%%%%%%%%%%%%%%%%%%%%%%%%%%%%%%%%%%%%%%%%%%%%%%%%%%%

\begin{figure}
    \centering
    \includegraphics[height=2in]{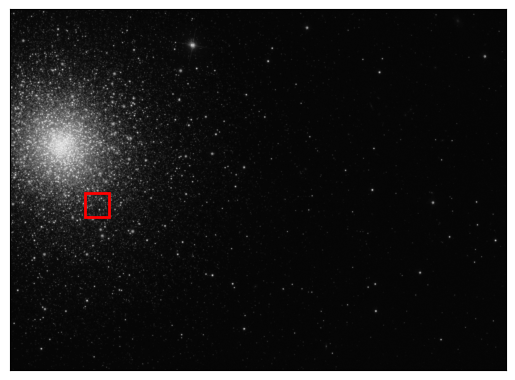}
    \hspace{1em}
    \includegraphics[height=2in]{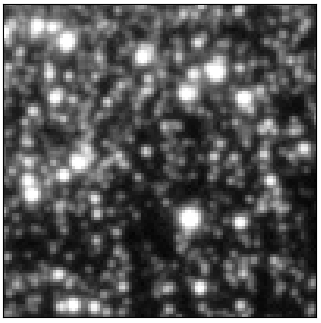}
    \caption{Left) An SDSS image of the M2 globular cluster (run 2583, camcol 2, field 136). The red box indicates the 100 $\times$ 100-pixel region we study. Right) An enlarged version of the studied region, rendered with arcsinh transformation and clipping.}
    \label{fig:m2-frame}
\end{figure}

\begin{figure}[t]
    \centering
    \includegraphics[width=1.0\textwidth]{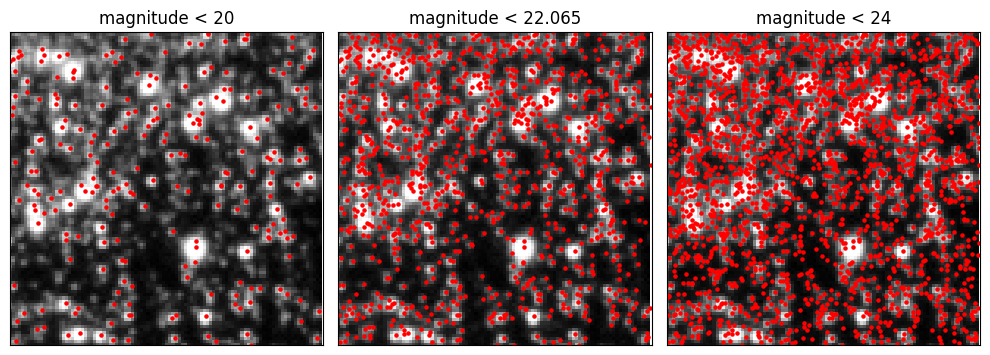}
    \caption{Three subsets of the HST catalog, each overlaid on the studied image (rendered using inverse sinisoidal sine normalization and clipping). From left to right, the panels use magnitude cuts that leave 133, 1114, and 2801 stars, respectively.}
    \label{fig:m2-hst-detections}
\end{figure}

The Messier 2 (M2) globular cluster is a gravitationally bound group of stars in the constellation Aquarius, located approximately 55,000 light years from Earth.
The density of objects in M2 is hundreds of times greater than in case study 1. In this setting, most objects visually overlap with many other objects. See \cref{fig:m2-frame,fig:m2-hst-detections}.
\citet{portillo2017improved} were the first to catalog images of M2 from the Sloan Digital Sky Survey (SDSS).
These authors used MCMC to perform this feat.

\subsubsection{Data generation}
\label{m2_data_gen}

% We assume the sources are uniformly distributed spatially and that their fluxes follow a power-law distribution with index 2 between fmin, which is fixed, and fmax, which is unknown. The reciprocal of fmax is given a uniform prior on-( )f0, bound 1 , where fbound is slightly dimmer than the brightest point source in the image, forcing fmax to be at least as bright as the brightest point source
% an exponential prior on N
% We set fmin to correspond to the SDSS 95% completeness limit in a sparse field (≈22mag in r band), expecting the dimmest significant source in the crowded field to be brighter than this limit.

% We set the power law slope α = 0.5 and use a standard Gaussian for the color prior (μc = 0, σ2c = 1), as in Feder et al.[2020].
% Rather than having a hierarchical structure, the prior parameters are fixed in this model: our goal is to produce aposterior on catalogs for a specific image, not to model the population over many images.
% Over a range of μ between 0.08 (corresponding to an prior average of 800 stars on a 100 × 100 image) and 0.20 (corresponding to 2000 stars) the F1-score remains steady between 0.49 and 0.51 (Table A.1). Similar robustness in F1 hold when α varies between 0.25 and 1.0 (Table A.2).

%    n_tiles_h: 56
%    n_tiles_w: 56
%    max_sources: 6
%    mean_sources: 0.8029160898574715
%    star_flux:
%        exponent: 0.9859821185389767
%        truncation: 5685.588160703261
%        loc: -1.162430157551662
%        scale: 1.4137911256506595

For this case study, our prior distribution on star flux, which was denoted by a generic $F$ in \cref{sec:model}, resembles that of previous work studying M2. Like \citet{portillo2017improved}, \citet{feder2020multiband}, and \citet{liu2023variational}, we model star flux with a truncated Pareto distribution prior. Whereas  \citet{portillo2017improved} sets the truncation level based on the brightest point source in the image, we fit the parameters of our scientific prior to catalog data for regions nearby but disjoint from the studied region. Our Pareto shape parameter was $\alpha = 0.98$, whereas \citet{portillo2017improved} used $\alpha = 2$ and \citet{liu2023variational} used $\alpha = 0.5$. \citet[][Appendix D]{liu2023variational} explored $\alpha \in [0.25, 1]$, finding low sensitivity to $\alpha$ in this range. As in these previous works, our truncated Pareto prior has a parameterized minimum flux value, which we set to 24 magnitude: stars fainter than 24 magnitude are not just individually undetectable in SDSS---they are so faint that even in aggregate they are expected to contribute few photoelectrons to images, according to the CDF of a Pareto distribution with our setting of $\alpha$.

For the prior on the number of objects $S$, like \citet{liu2023variational}, we model $S$ as Poisson, with a rate set based on prior catalogs of nearby regions. In contrast, \citet{portillo2017improved} and \citet{feder2020multiband} use a geometric distribution as a prior on the number of objects to penalize the detection of faint stars and thus reduce the number of objects detected, while acknowledging that Poisson may be more consistent with a proper Bayesian treatment of the problem.

We generate 246,784 images that are each 112$\times$112 pixels, as well as the corresponding catalogs.
In expectation, 2518 stars are rendered per image.
The majority of these stars are essentially undetectable even without blending, given the background noise level; but they are bright enough in aggregate to affect the image.

\subsubsection{Fitting the variational distribution}

We set the tile size to $2\times 2$ pixels. With this setting, in expectation, 61.9\% of the tiles contain no stars, 29.7\% contain exactly one star, 7.1\% contain exactly two stars, and 1.3\% contain three or more stars.
We set the number of detections per tile to $M = 2$.
This setting of $M$ precludes detecting 1.5\% of objects, which reside in tiles with three or more objects.
However, it is unlikely that we would have been able to detect many of these objects with useful precision regardless of $M$: such densely packed objects are difficult to distinguish in survey images.
% By setting $M$ to a relatively small value, we improve the computational efficiency and memory usage of our procedure, without missing many objects.
% We additionally marginalize over catalog entries for stars that are fainter than 22.565 magnitude by filtering them from the ground-truth catalogs, as described in \cref{sec:of-interest} and \cref{sec:npe-advantages} .

We consider two settings of $f_\mathcal{N}$ (cf. \cref{sec:amortized}) to assess how the type of information provided to the neighborhood network about objects affects performance. In the first setting, which is more minimalist, the neighborhood network considers only the positions of the objects within rank-$k$ tiles and the number of objects in pre-rank-$k$ tiles. In the second setting, we additionally include the fluxes of these objects. To simplify the presentation, the remainder of \cref{sec:m2} presents results solely for the first setting of $f_\mathcal{N}$, which produced better-calibrated posterior approximations. Results for the second setting appear in Appendix~E. The detrimental effect of providing more contextual information is a surprising result and is the focus of the discussion in \cref{sec:npe-discussion}. In brief, richer conditioning information may accentuate the distributional shift from training time, when ground-truth properties of neighboring objects serve as predictors (along with the image), to inference time, when predicted properties of neighboring objects are provided instead. This distributional shift is specific to NPE with hierarchical variational distributions, and it occurs even when the training and test data come from the same distribution.

% We used a batch size of 10 images. Before being fed to the inference network, images are preprocessed and normalized through a variety of nonlinear transformations, including contrast-limited adaptive histogram normalization and inverse hyperbolic sinisoidal transformations with normalization based on quantiles of the pixel intensities. Each transformation of the raw image becomes a separate channel in the input to our CNN. These transformations appear to be an important step as standard convolutional networks have difficulty with the high variability in pixel intensities in astronomical images, which for our simulated M2 data vary by three orders of magnitude (from as few as 641 photoelectrons in one pixel to more than $10^6$).

\subsubsection{Simulations to assess aggregate predictive performance}
\label{m2_synth_prediction}

\begin{table}
\centering
\begin{tabular}{c|c|c|c|c}
           & log-likelihood & precision & recall & F1 score \\
           \hline
1 rank & -4,216,070 (15,209)  & 0.600 (0.0003) & 0.544 (0.0002)  & 0.571 (0.0002) \\
4 ranks & -3,454,335 (15,554)  & 0.627 (0.0003) & 0.555 (0.0002) & 0.589 (0.0002) \\
\end{tabular}

\caption{BLISS performance with synthetic data that mimics SDSS images of the M2 globular cluster.
The standard error appears in parentheses.
\label{tab:ind_vs_checkerboard}
}
\end{table}

\cref{tab:ind_vs_checkerboard} reports on BLISS's performance with held-out synthetic data sampled from the same distribution as the training data.
As the number of ranks increases, the log-likelihood of ground-truth catalogs under the fitted variational distribution increases. Although the standard deviations reported in parentheses are similar in magnitude to the differences between log-likelihoods, these differences are nevertheless highly repeatable, as they are computed using the same held-out data. With more ranks, more information is available for making inferences.

Next, we assess whether increasing the number of ranks $K$ improves precision and recall for point estimates based on the mode of the fitted variational distribution. 
\cref{tab:ind_vs_checkerboard} shows that by using four ranks instead of one, we attain 1.1\% higher recall and 2.7\% higher precision.
% As in \cref{sec:sdss_field_synth}, the definitions of precision and recall are generalizations of the standard notions of precision and recall that apply to spatially dependent data.
Here, recall is the proportion of brightest 1114 objects detected (i.e., those brighter than 22.065 magnitude), where an object is detected if its position is inferred to within 0.5 pixels (in Euclidean distance) and its flux is inferred to within 0.5 magnitude.
Precision is the proportion of inferred objects with magnitude less than 22.065 that match an object in the ground truth catalog, where a match is defined as above.
These definitions of precision and recall, including the distance and flux difference thresholds, are selected for compatibility with previously published results.

% Precision and recall, though interpretable, do not tell us much about the calibration of the variational distribution.

\subsubsection{Simulation-based calibration}
\label{sec:m2_synth}

\begin{figure}[t]
    \centering
    \includegraphics[width=0.42\linewidth]{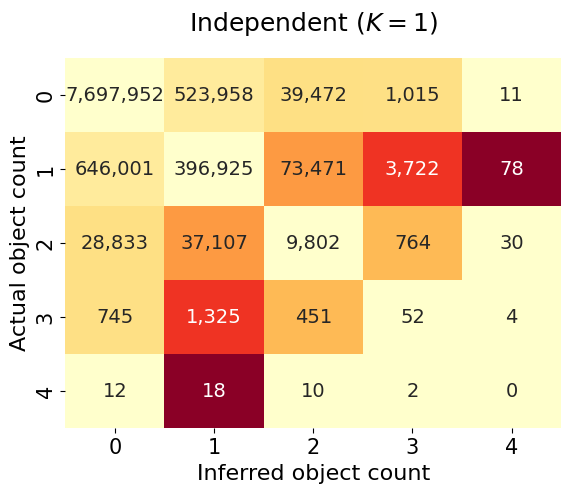}
    \hspace{20px}
    \includegraphics[width=0.42\linewidth]{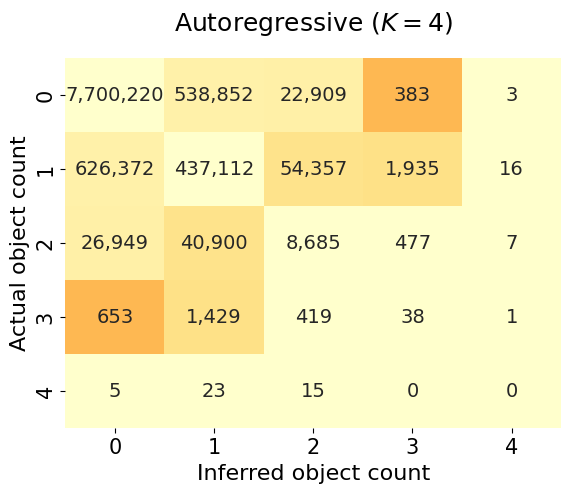}
    \caption{Confusion matrices for samples of catalogs from the BLISS variational distribution for synthetic images that mimic the SDSS image of the M2 globular cluster. For each catalog, whether inferred or actual, object count is the number of objects in a thin region along tile boundaries.
    Each cell's color indicates the difference between the entries of the corresponding transpose pair as a proportion of the smaller of the two entries (for transpose pairs with at least 100 observations in each cell).
    }
    \label{fig:m2_confusion_matrices}
\end{figure}

We use the procedure explained in \cref{sec:sdss_field_synth} to assess the calibration of our variational distribution. This time the events we consider involve the number of objects in a region that straddles the border of pairs of adjacent tiles. More specifically, we consider the number of objects in a region of either $2 \times 0.25$ or $0.25 \times 2$ pixels that overlaps with two $2\times{}2$-pixel tiles, with half of its area in each.

Figure~\ref{fig:m2_confusion_matrices} provides the resulting confusion matrices.
The variational distribution with $K=1$ exhibits poor calibration. It is twice as likely to infer that one object is two than to infer that two objects are one. It is nearly twice as likely to infer that one object is three as it is to infer the reverse. It is nearly twice as likely to infer that two objects are three than vice versa.
The variational distribution with $K=4$ overstates the number of objects somewhat too, but to a much lesser extent.

\subsubsection{Application to an SDSS image of M2}

We apply BLISS to the 100$\times$100-pixel region of an SDSS image of M2 shown in the right panel of \cref{fig:m2-frame} and, with overlaid HST catalogs, in \cref{fig:m2-hst-detections}. This region has also been imaged by the Hubble Space Telescope (HST), which, as a space-based telescope, produces sharper images than the SDSS telescope, which is ground-based. Ground-based telescopes, though not as sharp, typically cover a much larger region of the sky. Although exact ground truth is unknowable, the regions where both have coverage provide a convenient validation strategy: let space-based measurements serve as a proxy for ground truth. This is the strategy taken by \citet{portillo2017improved}, \citet{feder2020multiband}, and \citet{liu2023variational}, all of whom used this 100$\times$100-pixel image of M2 and the HST catalog as a proxy for ground truth. Using the same image and catalog facilitates comparison with these previous works.

% In \cref{fig:m2_reconstructions}, we examine the reconstruction error for a variational distribution with $K=4$, using the median of the inferred variational distribution as a point estimate of the catalog. The SDSS image, the HST reconstruction (that is, the reconstructed image based on our ``ground truth'' catalog), and the reconstruction based on the median of the BLISS variational distribution are visually almost indistinguishable. However, keep in mind that \cref{fig:m2_reconstructions} shows reconstructions with one particular arcsinh transformation; with different transformations, the agreement may not be so exact. \cref{fig:m2_reconstructions} serves mainly as a ``sanity check.'' In addition to confirming the reasonableness of checkerboard inference, the comparison between the HST reconstruction and the original SDSS image suggests that our model of the point spread function (PSF) and that our conversion of flux units between HST and SDSS are reasonable.

Next, we compare BLISS with three existing methods using the metrics and results reported in previous publications. The first, DAOPHOT~\citep{stetson1987daophot}, is widely used and explicitly designed for crowded starfields like M2. It is a non-probabilistic approach. The second, PCAT, was developed through two publications~\citep{portillo2017improved, feder2020multiband} and is based on MCMC.
The third, StarNet~\citep{liu2023variational}, uses variational inference with independent tiling.

\begin{table}[t]
\centering
\begin{tabular}{l|ccc|cc}
\toprule
& & & & \multicolumn{2}{c}{\#Stars} \\
     Method &   precision &  recall &   F1 score &  mean & (q-5\%, q-95\%)\\
\midrule
  DAOPHOT~\citep{stetson1987daophot} &  0.65 &  0.20 &      0.31 &     357 & -- \\
     PCAT~\citep{feder2020multiband} &  0.37 &  0.55 &      0.44 &    1672 & (1664, 1680)\\
 StarNet~\citep{liu2023variational}  &  0.48 &  0.53 &      0.50 &    1462 & (1430, 1497)\\
 BLISS (ours, with $K = 4$)          &  0.58 &  0.55 & 0.57 &  1171   & (1146, 1199)\\
\bottomrule
\end{tabular}
\caption{Performance on the SDSS image of the M2 globular cluster.
For probabilistic methods,
the ``\#stars" columns provide the posterior mean along with the 5th and 95th posterior percentiles
for the number of stars.
The number of stars in the ground-truth Hubble catalog is 1114. }
\label{tab:shootout}
\end{table}

\cref{tab:shootout} summarizes our results. BLISS makes a significant improvement in the F1 score.
At a recall rate slightly higher than that of StarNet and similar to that of PCAT, BLISS achieves 12\% higher precision than StarNet and 23\% higher precision than PCAT. All the probabilistic methods substantially outperform DAOPHOT, which has extremely low recall. Further, StarNet and PCAT made use of multiband images to attain these results, whereas BLISS used only the SDSS r-band image. (Properly modeling multiband SDSS images is not a trivial extension because different image bands are not pixel-aligned, so we defer it to future work.) In a sense, these results are even more noteworthy given that BLISS made use of less image data to form these estimates.

Our earlier results with synthetic data (\cref{m2_synth_prediction}) suggest that the autoregressive aspect of BLISS, assessed by comparing BLISS with $K = 1$ to BLISS with $K = 4$, contributes 1.6\% to the F1 score. 
Running both versions of BLISS ($K = 1$ and $K = 4$) on the real M2 data is consistent with this perspective, with $K = 4$ typically showing better performance, albeit with some inconsistency due to random initializations and the small sample size. (A difference of 1.6\% in the recall of 1114 stars would amount to just 18 stars.)
It follows that the majority of the 7\% F1 score improvement of BLISS over StarNet reported in \cref{tab:shootout} is due to changes other than the addition of autoregressive tiling, such as an improved neural network architecture (deeper, fully convolutional, and with more skip connections), a similar-but-not-identical prior distribution (cf. \cref{m2_data_gen}), and a more flexible variational approximation of the dependence between objects within the same tile.
% A variant of BLISS with one rank is included in \cref{tab:shootout} as an ablation benchmark to help explain why BLISS improves on StarNet. BLISS with one rank (i.e., without autoregressive tiling) also thoroughly outperforms StarNet. The relative performance of the one-rank and four-rank versions of BLISS here with real data is similar to what we observed with synthetic data in \cref{m2_synth_prediction}. This suggests that of the 7.6\% improvement in F1 score that the four-rank variant of BLISS makes on StarNet (which also lacks an autoregressive structure), 1\% is due to the addition of autoregressive tiling and 6.6\% is due to other improvements.

For probabilistic methods, the rightmost columns of \cref{tab:shootout} assess the calibration of the approximate posterior. The 90\% credible intervals for the number of objects show that StarNet and PCAT were not well calibrated, as the true number is 1114.
Our results show that miscalibration is less of an issue for BLISS than for StarNet and PCAT, but that the truth remains outside of the 90\% credible interval. Rerunning BLISS with $K = 1$ does not change this, showing that autoregressive tiling is not the cause. That this miscalibration did not occur with synthetic data suggests that miscalibration is due to model misspecification, which is to be expected to some degree. For example, our model of the PSF is likely imperfect. We examine the effect of model misspecification further in \cref{fig:mag-bin-performance}.
% Because autoregressive tiling is our focus here, we do not delve more deeply into the reasons BLISS was well calibrated in this respect whereas StarNet was not. %Informally, it appears to be influenced by more than one of the differences between the one-rank variant of BLISS and StarNet listed above, with the deeper architecture and the refitted prior as leading candidates.

\begin{figure}
    \centering
    \includegraphics[width=0.9\textwidth]{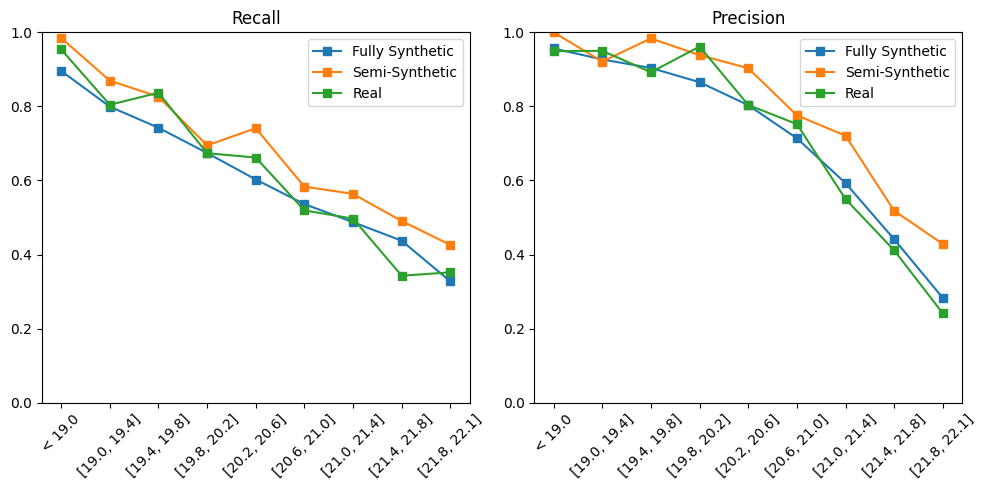}
    \caption{BLISS's performance (precision and recall) binned by object magnitude, with three types of data: real (i.e., the 100$\times$100-pixel SDSS subimage); semi-synthetic (pixel intensities drawn from the BLISS generative model, conditioned on the ``ground truth'' HST catalog); and synthetic (both catalog and image drawn from the BLISS model).}
    \label{fig:mag-bin-performance}
\end{figure}

\cref{fig:mag-bin-performance} shows BLISS performance as a function of star magnitude. It also compares BLISS's performance with the real SDSS data to its performance with synthetic and semi-synthetic data. The synthetic data is a sample from our generative model, whereas the semi-synthetic data is a conditional sample from our generative model of pixel intensities given the real M2 catalog.
The performance difference with semi-synthetic and real data suggests that the misspecification of the conditional likelihood is non-negligible, resulting in a 5--10\% reduction in precision and recall. There also appears to be non-negligible misspecification of the catalog prior, as precision and recall are significantly lower for the real data than for the semi-synthetic data.
This misspecification, though undesirable, is not so severe that it prevents us from achieving state-of-the-art performance.

%%%%%%%%%%%%%%%%%%%%%%%%%%%%%%%%%%%%%%%%%%%%%%%%%%%%%%
\section{Discussion} \label{sec:discussion}
%%%%%%%%%%%%%%%%%%%%%%%%%%%%%%%%%%%%%%%%%%%%%%%%%%%%%%

We proposed a novel amortized variational distribution for astronomical object detection with a spatially autoregressive structure and the desirable theoretical property that, by construction, its dependency structure matches that of the posterior.
Although our two case studies have quite different settings, in many ways, both demonstrated the value of autoregressive tiling.
Despite this success, our method has limitations due to its reliance on NPE (\cref{sec:npe-discussion}) and
the form of the variational distribution (\cref{sec:gp-discussion}).
Ultimately, probabilistic catalogs will need to be integrated into science workflows, which can be done in several ways that propagate uncertainty to downstream analyses (\cref{sec:asto-discussion}).

\subsection{The potential and pitfalls of neural posterior estimation}
\label{sec:npe-discussion}

This work is among the few to use NPE, which is likelihood-free, in a setting in which efficient likelihood evaluation is possible.
Our results suggest that the use of likelihood information is not essential to achieving good performance.

This work is also among the few to use NPE in the context of a hierarchical (e.g., autoregressive) variational distribution.
Our work identifies a drawback of NPE that is specific to fitting hierarchical variational distributions, which was surprisingly significant in our case study with the M2 globular cluster.
In this case study, we initially selected neighborhood network $f_\mathcal{N}$ to include as much detail as possible about nearby objects. However, doing so caused the networks to encounter quite different data during training (when ground truth is used) than during inference (when previously detected objects are used).
We term this issue ``exposure bias'' to connect it to a related issue of the same name in autoregressive text generation.

We worked around this problem by limiting the information that $f_\mathcal{N}$ considers, no longer considering the flux of objects, and, for neighboring objects in other tiles, no longer considering their positions within the tile. Excluding this information reduces not only the log-likelihood (which is not subject to exposure bias) but also the precision and recall of predictions based on the mode of our variational distribution (which is). In fact, the best F1 score reported in \cref{tab:shootout} was not our best. Including flux information increases it by 1\%. However, doing so degrades the calibration of the fitted variational distribution. Appendix~E contains a version of Figure~\ref{fig:m2_confusion_matrices} that was generated with richer conditioning information.

In contrast, exposure bias was not an issue in the first case study, suggesting that the higher object density in the second case study is part of the story. The higher object density required us to use smaller tiles in defining the variational distribution and to detect multiple objects in some tiles.
There were also many more tiles with detections in the second case study: with many detections, there are more ways for sampled catalogs (encountered during inference) to differ from the real catalogs (encountered during training).

For exposure bias in autoregressive text generation, adversarial penalties, such as professor forcing \citep{lamb2016professor}, are a widely used mitigation strategy, as discussed in Appendix~C. This may also be an effective mitigation strategy for NPE, as it would encourage the network to ignore the subtle differences between synthetic training data and real telescope images.

\subsection{Limitations of tile-based variational distributions}
\label{sec:gp-discussion}
Variational distributions are essentially never flexible enough to recover the exact posterior distribution, but sometimes they are nevertheless adequate for the task at hand. Other times, however, more exactness is needed. In these cases, the variational distribution can typically be made more flexible to remedy specific limitations. Below, we discuss two limitations of the proposed variational distribution. 

\paragraph*{Bounded number of detections per tile}
We modeled objects as either ``of interest'' or ``nuisance'' (\cref{sec:of-interest}). 
The number of objects of interest $M$ is bounded per tile, and this bound needs to be fairly small because the runtime has a factorial dependence on $M$ (cf. \cref{eq:rank-k-tiles-marginal-order}). 
For many applications, by making the tiles small (e.g., a few pixels in size), we can ensure that small $M$ suffices: only a few objects centered in such a small tile can be recovered with useful precision anyway.

If larger $M$ were needed, we could potentially get a workable lower bound on the log density of the variational distribution without explicitly considering all $M!$ permutations by using a bipartite matching algorithm. Replacing the expectation in \cref{eq:rank-k-tiles-marginal-order} with the most probable bipartite match between predictions and truth amounts to replacing a softmax with a ``hard'' max, with all probability masses and densities encoded on a logarithmic scale.

However, it is unsatisfying to truncate the support of the approximate posterior at all because under the prior the number of objects per image, and hence per tile (due to Poisson thinning), follows a Poisson distribution, which has unbounded support. Allowing the number of detections of interest $M$ to be random is one way forward. The challenges with random $M$ tend to be computational rather than conceptual: computer hardware, and especially GPUs, are most efficient when performing vectorized operations on tensors of fixed dimensions, for example, when all vectors can be padded to have finite length $M$.

\paragraph*{Computation required for long-range dependence}
Our variational distribution has the structure of a $K$-color checkerboard.
Larger $K$ requires more sequential computation but can recover dependence between detections that are more distant from each other.
Although there is, in general, a trade-off between computational expense and approximation fidelity, if objects are small in terms of spatial extent relative to tile size $T$, we do not have to choose: because of conditional independence within the posterior distribution, it suffices to use relatively small $K$ (cf. Appendix~B, Theorem~4).

A small fraction of galaxies, however, are large. Because $R$ is determined by the largest possible object, the theorem implies that a large $K$ is needed, which could be computationally burdensome.

A potential solution is to define the variational distribution in terms of multiple checkerboards, each with different tile sizes. 
The checkerboards with large tiles would be responsible for detecting larger objects, thus reducing the need for more colors in checkerboards with small tiles while still recovering long-range dependencies.
These checkerboards would themselves be ranked, perhaps to detect the largest objects first, thus adding additional autoregressive structure to our variational distribution.

% However, with multiple overlapping checkerboards, it could be ambiguous which checkerboard is ``assigned'' to detect a particular object, particularly for objects whose sizes are on the threshold between pairs of checkerboards. Introducing latent assignment variables into our formulation of the variational distribution could potentially account for this ambiguity, but then the exact log-likelihood in the variational distribution of catalogs likely would no longer be tractable, perhaps necessitating the use of a lower bound.

% Note that while it is often straightforward to specify a flexible variational distribution by describing procedures for sampling from it, during fitting, NPE requires more than sampling: NPE requires that we can efficiently compute the likelihood of arbitrary catalogs under the variational distribution.

% While the computation of \cref{eq:tiles} can be parallelized across tiles of the same rank, it cannot be further parallelized across the permutations of the $M$ potential objects within each tile, as they appear within a logarithm:
% \begin{align*}
%\nabla \log q(z_\ell \mid w_{<k})
%&= \nabla \log \sum_{\sigma} \left[ \prod_{i = 1}^M q(z_\ell^{[i]} \mid z_\ell^{[1]}, \ldots, z_\ell^{[<i]}, w_{\preceq\ell}) \right].
%\end{align*}
%Partly for this reason, we advocate using many small tiles (e.g., $2\times{}2$ pixels) rather than fewer larger ones. In this way, $M$ can also be kept small without compromising the quality of the posterior approximation.

\subsection{Implications for astronomical practice}
\label{sec:asto-discussion}

Historically, progress in observational astronomy has been driven by advances in telescope technology~\citep{ivezic2019lsst}, mediated by improvements in astronomical catalogs, which are the starting points for most analyses~\citep{schafer2015framework}. Each generation of telescopes provides higher-fidelity images of the Universe than the last. The Sloan Digital Sky Survey (SDSS), which began collecting images in 2000, featured an optical telescope equipped with a 2.5-meter primary mirror; the Dark Energy Survey (DES), which began in 2006, used a 4-meter mirror; and the Large Survey of Space and Time (LSST), which is scheduled to begin collecting data in 2025, will feature an 8.4-meter mirror. These mirrors weigh thousands of kilograms, and the corresponding telescopes cost hundreds of millions of dollars to construct.

By improving the accuracy of catalogs, the proposed method has the potential to make catalogs derived from today's telescopes more like the ones derived from the telescopes of the future, without the delay and the expense of constructing new telescopes.

However, transitioning from deterministic to probabilistic catalogs requires adjustment in how catalogs are used, as existing scientific pipelines, which typically infer properties of populations of astronomical objects~\citep{schafer2015framework}, are designed to receive deterministic catalogs as input.
For instance, redMaPPer~\citep{rykoff2014redmapper} uses deterministic catalogs to find and characterize galaxy clusters, which are made up of hundreds of nearby galaxies; CHIPPR~\citep{malz2022obtain} estimates redshift distributions from photo-Z estimates for individual galaxies; and the method of \cite{schneider2017probabilistic} estimates cosmological parameters from catalogs of observed ellipticities.

One simple way to use probabilistic catalogs is to summarize them with a point estimate that is provided as input to existing scientific pipelines.
This practice fails to propagate uncertainty to downstream analyses, although it may not fully negate the value of probabilistic cataloging: the point estimates have the potential to be more accurate than algorithmically derived point estimates.

An alternative approach is to input samples of catalogs from our posterior approximation into existing scientific pipelines and rerun the pipeline once for each sample to produce an estimate of a population-level quantity. Collectively, the runs would give us a distribution of estimates; the variation in these estimates would represent uncertainty.
This approach is appealing in that it allows reuse of existing software pipelines for population-level analysis, while still accounting for catalog uncertainty.
One concern, however, is that seemingly minor miscalibrations in our posterior approximation for each object, which are inevitable to some extent, could compound in the population-level analysis, leading to significant miscalibration in the distribution of estimates of population-level parameters.

A third approach to population-level analysis is to modify our simulator to directly model population-level parameters, rather than relying on existing scientific pipelines. For instance, we can generate astronomical images that exhibit galaxy clusters, weak lensing, and strong lensing.
By incorporating models of these phenomena into our simulator, it becomes possible to directly estimate these quantities using neural posterior estimation. 
% We anticipate that early adopters will likely be analyses where detection uncertainty is currently the dominant systematic (e.g., cluster cosmology in crowded fields), with broader adoption following as the community develops standard tools for uncertainty propagation.
In principle, the simulation could extend all the way to sampling the parameters of the $\Lambda$CDM model of cosmology, which characterizes the origin and evolution of the Universe.
The current work lays the foundation for such extensions.

% ``produce a K-catalogue sampling in `catalogue space' that samples a posterior probability distribution of catalogues given the data'' \citep{hogg2010telescopes}

%a step towards unlocking the deep learning revolution for astronomical image analysis

%%%%%%%%%%%%%%%%%%%%%%%%%%%%%%%%%%%%%%%%%%%%%%
%% Appendix---Please move all appendices to %%
%% a Supplementary file.                    %%
%%%%%%%%%%%%%%%%%%%%%%%%%%%%%%%%%%%%%%%%%%%%%%
%% Support information, if any,             %%
%% should be provided in the                %%
%% Acknowledgements section.                %%
%%%%%%%%%%%%%%%%%%%%%%%%%%%%%%%%%%%%%%%%%%%%%%
%\begin{acks}[Acknowledgments]
% The authors would like to thank ...
%\end{acks}
%%%%%%%%%%%%%%%%%%%%%%%%%%%%%%%%%%%%%%%%%%%%%%
%% Funding information, if any,             %%
%% should be provided in the                %%
%% funding section.                         %%
%%%%%%%%%%%%%%%%%%%%%%%%%%%%%%%%%%%%%%%%%%%%%%
\begin{funding}
This material is based on work supported by the National Science Foundation under Grant No. 2209720 and the U.S. Department of Energy, Office of Science, Office of High Energy Physics under Award Number DE-SC0023714.
\end{funding}

\vspace{12pt}
\textbf{Software: } The software described in this article is available from \url{https://github.com/prob-ml/bliss}. Code specifically for reproducing the results in this article is contained within the \texttt{case\_studies/spatial\_tiling} directory.

%%  imsart-nameyear.bst  will be used to                   %%
%%  MR numbers will be added by VTeX.                      %%
%%                                                         %%
%%  Use \cite{...} to cite references in text.             %%
%%                                                         %%
%%%%%%%%%%%%%%%%%%%%%%%%%%%%%%%%%%%%%%%%%%%%%%%%%%%%%%%%%%%%%
\bibliographystyle{imsart-nameyear} % Style BST file
\bibliography{references}       % Bibliography file (usually '*.bib')

\clearpage
\appendix
\counterwithin{figure}{section}
\renewcommand{\thefigure}{\thesection.\arabic{figure}}

\section{Properties of the generative model}
\label{sec:posterior-ind}

Recall the following notation from Sections~2 and 3 of the main text.
The data is an $H\times{}W$-pixel image $x$.
The natural number $T$ is a user-defined tile side length in pixels, typically set to 2 or 4.
The constants $H' = H / T$ and $W' = W / T$, assumed to be natural numbers, are the image dimensions measured in tiles.
The set of all tile indices is $\Omega \coloneqq \{1,\ldots,H'\}\times\{1,\ldots,W'\}$.
For a subset of the tile indices $L \subset \Omega$, we denote the corresponding latent random variables $y_L \coloneqq \{y_i : i \in L\}$.
We refer to the elements of $y_\Omega$ as the tile latent variables.

With the change of variables proposed in Section 3.1, due to the thinning property of the Poisson distribution \citep[Section 3.7.2]{durrett2019probability}, our catalog prior factorizes as
\begin{align}
    p(y) = \prod_{(h, w) \in \Omega} p(y_{h, w}).
\end{align}

Nevertheless, posterior independence between pairs of tile latent variables, even those corresponding to distant tiles, does not generally hold for this model: random variables are coupled through the chains of potential objects between them. However, conditional independence does hold for tile latent variables that are sufficiently far apart, given the other tile latent variables, as the following proposition shows.

\begin{restatable}[]{myprop}{posteriorind}
\label{posterior-ind}
Let $A \subset \Omega$ and $C \subset \Omega$ be subsets of the tile indices such that for all $a \in A$ and $c \in C$, $\|a - c\|_\infty \ge 2R/T + 1$. Let $B \coloneqq \Omega \setminus (A \cup C)$. Then, in the posterior,
\begin{align*}
    y_A \ind y_C \mid y_B.
\end{align*}
\end{restatable}
\begin{proof}
\label{posterior-ind-proof}
Let $A^+$ (resp. $C^+$) refer to the tile indices contained in $A$ (resp. $C$) as well as those of a tile within $R$ of any tile in $A$ (resp. $C$). Let $B^- \coloneqq \Omega \setminus (A^+ \cup C^+)$ refer to the tile indices contained in neither $A^+$ nor $C^+$. 
Observe that $y = y_{A^+} \bigsqcup y_{B^-} \bigsqcup y_{C^+}$ and $y = y_A \bigsqcup y_B \bigsqcup y_C$, where $\bigsqcup$ denotes disjoint union. 

Let
\begin{align*}
\varphi_{AB} \coloneqq \prod_{(h,w) \in A} p(y_{h,w}) \prod_{(h, w) \in A^+} p(x_{h,w} \mid y_A, y_B),
\end{align*}
\begin{align*}
\varphi_{B} \coloneqq \prod_{(h,w) \in B} p(y_{h,w}) \prod_{(h, w) \in B^-} p(x_{h,w} \mid y_B),
\end{align*}
and
\begin{align*}
\varphi_{BC} \coloneqq \prod_{(h,w) \in C} p(y_{h,w}) \prod_{(h, w) \in C^+} p(x_{h,w} \mid y_B, y_C).
\end{align*}
Because $y$ follows a spatial Poisson process, the density of the joint distribution may be expressed as a product over tiles:
\begin{align*}
    p(x, y)
    &= \prod_{h,w} p(y_{h,w}) p(x_{h,w} \mid \{y_{i,j} : |i -h| \le \lceil R \rceil \wedge |j - w| \le \lceil R \rceil\})\\
    &= \varphi_{AB} \varphi_{B} \varphi_{BC}.
\end{align*}
Using this equality, it follows that
\begin{align*}
    p(y_A, y_C \mid y_B, x)
    &= \frac{p(x, y)}{\int p(x, y_B)}\\
    &= \frac{p(x, y)}{\int p(x, y) dy_A d y_C}\\
    &= \frac{\varphi_{AB} \cancel{\varphi_{B}} \varphi_{C}}{\left[ \int \varphi_{AB} dy_A\right] \cancel{\varphi_{B}} \left[ \int \varphi_{BC} dy_C \right]}\\
%    &= \frac{\varphi_{AB}}{\left[ \int \varphi_{AB} dy_A\right]}  \frac{\varphi_{BC}}{\left[ \int \varphi_{BC} dy_C\right]}\\
    &= \frac{\varphi_{AB}\varphi_{B}\left[ \int \varphi_{BC} dy_C\right]}{\left[ \int \varphi_{AB} dy_A\right]\varphi_{B}\left[ \int \varphi_{BC} dy_C\right]}
    \frac{\left[ \int \varphi_{AB} dy_A\right] \varphi_{B} \varphi_{BC}}{\left[ \int \varphi_{AB} dy_A\right] \varphi_{B} \left[ \int \varphi_{BC} dy_C\right]}\\
    &= \frac{p(y_A, y_B, x)}{p(y_B, x)} \frac{p(y_C, y_B, x)}{p(y_B, x)}\\
    & = p(y_A \mid y_B, x) p(y_C \mid y_B, x).
\end{align*}
The proposition follows.
\end{proof}

%The proposition holds because $y_A$ and $y_C$ are d-separated by $y_B$. A detailed proof is given in \cref{posterior-ind-proof}. 

\newpage
\section{Shared structure between the model and the variational distribution}
\label{sec:blanket}
The variational family can be configured to match the independence assertions of the posterior distribution. To show this, we represent the independence assertions of both the posterior distribution and the variational distribution as Bayesian networks (BN), which are also known as directed graphical models.
In a BN, random variables are represented as nodes, potential conditional dependencies are represented by directed edges, and conditional independencies are denoted by the absence of edges \citep{koller2009probabilistic}.

The following notation is useful in expressing the BNs. For tile index $\ell \in \Omega$ and radius $r \in \mathbb R$, let the nearby tile variables $z_\ell^{r} \coloneqq \{z_{\ell'} : \|\ell - \ell'\|_\infty \le r, \, \, \forall \ell' \in \Omega\}$.
For a generic subset of tile indices $L \subset \Omega$, the corresponding latent random variables are $z_L \coloneqq \{z_\ell : \ell \in L\}$.
For any pixel index $a = (h, w)$, let $x_\ell \coloneqq x_{h,w}$.
Let $\Upsilon \coloneqq \{1, \ldots, H\} \times \{1, \ldots, W \}$.
For a generic subset of pixel indices $A \subset \Upsilon$, the corresponding pixels are $x_A \coloneqq \{x_a : a \in A\}$.
For $a \in \Upsilon$ and $r \in \mathbb R$, let the nearby pixels $x_a^{r} \coloneqq \{x_{a'} : \|a - a'\|_\infty \le r, \, \, \forall a' \in \Upsilon\}$.

\paragraph*{The variational distribution BN}
Using this new notation, the proposed variational distribution can be expressed as
\begin{align}
    q(z \mid x) = \prod_{k = 1}^K \prod_{\ell \in C_k} q(z_\ell \mid z_{<k} \cap z_{\ell}^{r_{\mathcal N}}, x_{(\ell - 1)T + 1}^{r_\mathcal X}).
    \label{eq:overall_q}
\end{align}
This distribution corresponds to the Bayesian network defined by a graph $G_q$ with vertices
\begin{align*}
    V_q \coloneqq \{x_a : a \in \Upsilon\} \cup \{z_\ell : \ell \in \Omega \}
\end{align*}
and directed edges
\begin{align*}
    E_q \coloneqq &\{(x_a, z_\ell) : \|a - (T(\ell - 1) + 1)\|_\infty \le r_{\mathcal X}, \,\, \forall  a \in \Upsilon,\,\, \forall \ell \in \Omega\} \\
    &\cup \{(z_\ell, z_{\ell'}) : \| \ell - \ell' \|_\infty \le r_{\mathcal N} \wedge \Psi(\ell) < \Psi(\ell'), \,\, \forall \ell \in \Omega,\,\, \forall \ell' \in \Omega\}.
\end{align*}

As a conditional distribution, our variational distribution does not make assertions about the dependencies or independencies among the variables conditioned on (i.e., the pixel intensities).
By convention, the exclusion of edges from $E_q$ between pixel intensities does not denote conditional independence; rather, the relationship among conditioned-on variables is simply beyond the scope of the graph.

\paragraph*{The generative model BN}

Our generative model, expressed in the variables of Section~3 of the main text, is $p(z, e, x)$, which has Bayesian network structure $G_p$ with vertices
\begin{align*}
    V_p \coloneqq\{x_a : a \in \Upsilon\} \cup \{z_\ell : \ell \in \Omega \} \cup \{e_\ell : \ell \in \Omega \}
\end{align*}
and directed edges
\begin{align*}
    E_p \coloneqq &\{(z_\ell, x_a) : \|a - (T(\ell - 1) + 1)\|_\infty \le \lceil R \rceil, \,\, \forall  a \in \Upsilon,\,\, \forall \ell \in \Omega\}\\
    &\cup \{(e_\ell, x_a) : \|a - (T(\ell - 1) + 1)\|_\infty \le \lceil R \rceil, \,\, \forall  a \in \Upsilon,\,\, \forall \ell \in \Omega\}\\
    &\cup \{(z_\ell, e_\ell) : \,\, \forall \ell \in \Omega\}.
\end{align*}
In relating this Bayesian network to our earlier model description, we use the thinning property of the Poisson distribution \citep[Section 3.7.2]{durrett2019probability} to express a draw of the image prior as a draw for each of the $H'\times{}W'$ tiles.

The edges of the form $(z_\ell, e_\ell)$ are included in $E_p$ because of the flexibility with which objects can be classified as ``of interest'' or not (see Section 3.2 of the main text).
If our interest in objects were solely due to their fluxes in relation to a pre-defined threshold, these edges would not be needed.
On the other hand, if some objects in tile $\ell$ are nuisance objects because tile $\ell$ contains more than $M$ brighter objects, then there is a dependence between the sets of interesting objects $z_\ell$ and the nuisance objects $e_\ell$.

\cref{fig:bn} gives an example of Bayesian networks for the generative model and the variational distribution.

\paragraph*{Comparing independence assertions}
We now compare the independence assertions of $q(z \mid x)$, as encoded by $G_q$, and $p(z \mid x)$, as encoded by $G_p$.
Only $G_p$ contains nodes for the nuisance objects,  $\{e_\ell : \ell \in \Omega\}$, which need to be marginalized over to conduct the comparison.
Unfortunately, performing variable elimination of the nuisance objects, e.g., via the sum-product algorithm~\citep{koller2009probabilistic}, introduces complex dependence among the remaining variables, which complicates analysis.
Fortunately, in the prior, for any $\ell \in \Omega$, the dependence between $z_\ell$ and $e_\ell$ becomes arbitrarily weak as $M$ increases.
Because we intend for the tile size to be small (e.g., 4 pixels), $M$ is typically large enough that the large $M$ setting is an interesting one.
Hence, we consider the setting in which $G_p^-$, the graph obtained by removing the edges of the form $(z_\ell, e_\ell)$ from $G_p$, is a Bayesian network of $p(z, e, x)$.
For $G_p^-$, variable elimination of the nuisance objects introduces dependencies among the pixels, but it does not link the remaining latent variables.

Our result that follows relies on the concept of a minimal I-map \citep{koller2009probabilistic}, which we first define.

\begin{figure}
    \centering
    \includegraphics[width=0.3\linewidth]{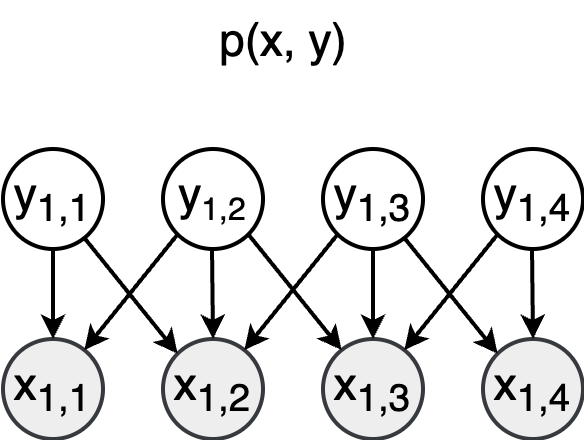}
    \hfill
    \includegraphics[width=0.3\linewidth]{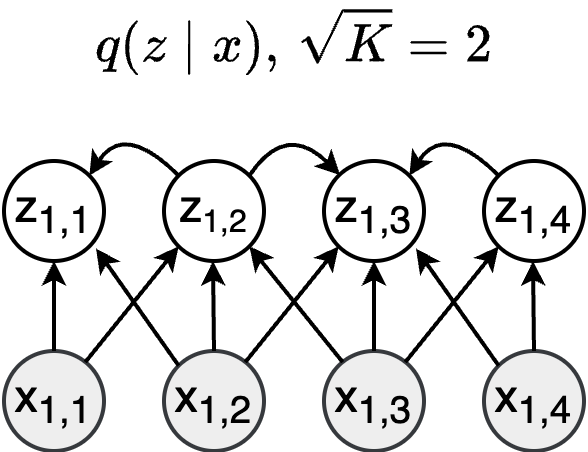}
    \hfill
    \includegraphics[width=0.3\linewidth]{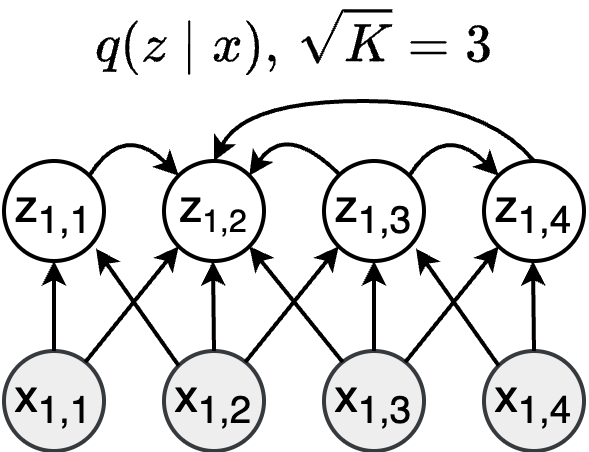}
    \caption{The Bayesian networks for a 1$\times$4-pixel image with tile size $T=1$ and object radius $R=1$. (Left) The generative model. (Center) The variational distribution with $K=4$, $r_{\mathcal X} = 1$, and $r_{\mathcal N} = 2$. (Right) The same variational distribution except with $K = 9$. Now, with an assumption, the variational distribution is minimally faithful stochastic inverse of the model.
    }
    \label{fig:bn}
\end{figure}

\begin{definition}
Let $G$ be a Bayesian network structure with a set of independencies $I(G)$. 
Let $P$ be a distribution and let $I(P)$ be the set of independence assertions that hold in $P$.
We say that $G$ is an I-map for $P$ if $I(K) \subset I(P)$.
\end{definition}

Note that if $P$ is a conditional distribution, it can make conditional independence assertions that involve the variables conditioned on, but not those that exclusively pertain to the variables conditioned on.

\begin{definition}
Let $G$ be a Bayesian network structure and let $P$ be a distribution.
$G$ is a minimal I-map for $P$ if it is an I-map for $P$ and
if the removal of any edge from $G$ renders it not an I-map.
\end{definition}

Recall that $G_q$ is parameterized by the number of ranks $K$ and the sizes of the receptive fields, $r_{\mathcal N}$ and $r_{\mathcal X}$. If these values are large enough, $G_q$ is an I-map for the posterior distribution. The following result provides values of $K$, $r_{\mathcal N}$, and $r_{\mathcal X}$ that are no larger than necessary.

\begin{restatable}[]{theorem}{qisfaithful}
\label{q-is-faithful}
Suppose that $G_p^-$ is an I-map for $p(x, z, e)$. Further, suppose that
$\sqrt{K} = \lceil 2R / T + 1 \rceil$, $r_{\mathcal X} = \lceil R \rceil$, and $r_{\mathcal N} = \lceil 2R / T \rceil$.
Then, $G_q$ is a minimal I-map for $p(z \mid x)$.
\end{restatable}
\begin{proof}
\label{q-is-faithful-proof}
Recall that $G_p^-$ has edges
\begin{align*}
    E_p^- \coloneqq &\{(z_\ell, x_a) : \|a - (T(\ell - 1) + 1)\|_\infty \le \lceil R \rceil, \,\, \forall  a \in \Upsilon,\,\, \forall \ell \in \Omega\}\\
    &\cup \{(e_\ell, x_a) : \|a - (T(\ell - 1) + 1)\|_\infty \le \lceil R \rceil, \,\, \forall  a \in \Upsilon,\,\, \forall \ell \in \Omega\}.
\end{align*}
Performing variable elimination on $G_p^-$ to marginalize over the nuisance objects 
exclusively adds edges between the pixel intensities, which are observed random variables, yielding the graph $G_z$ with vertices $V_q$ and edges
\begin{align*}
    E_z \coloneqq &\{(z_\ell, x_a) : \|a - (T(\ell - 1) + 1)\|_\infty \le \lceil R \rceil, \,\, \forall  a \in \Upsilon,\,\, \forall \ell \in \Omega\}\\
    &\cup \{(x_a, x_a') : a \preceq a', \,\, \forall \,\, a, a' \in \Upsilon \},
\end{align*}
where $\preceq$ denotes comparison under an arbitrary but fixed order.
The precise order is inconsequential for what follows.

Now, to see that $G_q$ is an I-map for $p(z \mid x)$, observe that
\begin{align}
    p(z \mid x)
    &= \prod_{k=1}^K p(z_k \mid z_{<k}, x) \label{eq:pzx_color}\\
    &= \prod_{k=1}^K \prod_{\ell \in C_k} p(z_\ell \mid z_{<k}, x) \label{eq:pzx_ind_within_color}\\
    &= \prod_{k=1}^K \prod_{\ell \in C_k} p(z_\ell \mid z_{<k} \cap z_{\ell}^{\lceil 2R/T \rceil}, x) \label{eq:pzx_coparents} \\
    &= \prod_{k=1}^K \prod_{\ell \in C_k} p(z_\ell \mid z_{<k} \cap z_{\ell}^{\lceil 2R/T \rceil}, x_{(\ell - 1)T + 1}^{\lceil R \rceil}) \label{eq:pzx_children}\\
    &= \prod_{k=1}^K \prod_{\ell \in C_k} p(z_\ell \mid z_{<k} \cap z_{\ell}^{r_{\mathcal N}}, x_{(\ell - 1)T + 1}^{r_{\mathcal X}}) \label{eq:pzx_subs}.
\end{align}
Here, \cref{eq:pzx_color} follows from the generalized product rule.
\cref{eq:pzx_ind_within_color} follows from our assumption that $\sqrt{K} > 2R/T$, which implies that nearby tiles, which have a common child in $G_z$, have different ranks.
\cref{eq:pzx_coparents} follows because $z_{\ell}^{\lceil 2R/T \rceil}$ contains all the tile latent variables that are co-parents of $z_\ell$ in $G_z$, and thus the conditioned-on latent variables constitute a Markov blanket.
\cref{eq:pzx_children} follows because the removed pixel intensities, i.e., those not in $x_{(\ell - 1)T + 1}^{\lceil R \rceil}$, are not children of $z_\ell$ in $G_z$.
\cref{eq:pzx_subs} follows by substitution.
We have thus expressed $p(z \mid x)$ in precisely the form of $q(z \mid x)$ in \cref{eq:overall_q},
whose factorization defined the Bayesian network $G_q$, thus establishing that $G_q$ is an I-map for $p(z \mid x)$.

To further establish that $G_q$ is minimal, consider removing an arbitrary edge from $E_q$. 
If the edge connects a pixel intensity $x_a$ to a tile latent variable $z_\ell$, then its removal would assert that $x_a$ and $z_\ell$ are independent given all other random variables; but the reversed edge, $(z_\ell, x_a)$, is contained in $G_z$, so this independence assertion does not hold for $p(z \mid x)$.
Alternatively, if the removed edge connects tile latent variables $z_\ell$ and $z_{\ell'}$,
then $\| \ell - \ell' \|_\infty \le r_{\mathcal N} = \lceil 2R / T \rceil$.
Thus, in $G_z$, $z_\ell$ and $z_{\ell'}$ are co-parents of one or more pixels, for example, $x_a$ where $a = \lceil T(\ell + \ell') / 2 \rceil$,
which implies that removing the edge renders $G_q$ not an I-map for $p(z \mid x)$.
\end{proof}

%A proof is given in \cref{q-is-faithful-proof}.
This result establishes the properties that make this variational family well suited to the problem; namely, that the conditional independence assertions of the variational distribution, given particular choices of $K$ and particular convolutional neural network architectures, are shared by the generative model.

\paragraph*{Implications}

This result also has implications for setting the number of ranks $K$ and the architecture of the inference network, which in turn determines the receptive field sizes $r_{\mathcal X}$ and $r_{\mathcal N}$, and the tile size $T$.
In particular, it is inadvisable to set $K$, $r_{\mathcal X}$, and $r_{\mathcal N}$ to larger values than prescribed by \cref{q-is-faithful}; doing so would give the variational distribution flexibility that is not needed, and thus would place the burden of learning to ignore irrelevant inputs on the fitting procedure, which may slow optimization and decrease the quality of the solution found in a feasible duration.

There is a drawback to setting $K$, $r_{\mathcal X}$, and $r_{\mathcal N}$ to smaller values than prescribed by \cref{q-is-faithful}. However, doing so may nevertheless be worthwhile: severing some conditional dependencies that could be present in the posterior lets us trade fidelity for computational efficiency, which is a common practice in VI. Neighboring tiles tend to have the strongest dependence; conditional dependence will be much weaker once immediate neighbors are conditioned on.

Another implication is that the number of ranks $K$ needed grows with the object size $R$ rather than the image size ($H \times W$), which is critical for processing high-resolution images. By keeping $K$ small, we can parallelize sampling across GPU cores instead of requiring lengthy sequential computation.

The $K$ required also varies with the tile size $T$.
With larger tiles, sampling the variational distribution requires fewer sequential steps. However, larger tiles are expected to contain more objects, making the distribution of their contents more difficult to infer and sample.

\newpage
\section{Comparison of ELBO-based VI and NPE}
\label{sec:npe-elbo-comparison}

Neural posterior estimation (NPE) is not a widely appreciated approach to Bayesian inference, and there appear to be no articles providing comprehensive guidance on using it.
Because NPE is central to our method, we therefore discuss its advantages (\cref{sec:npe-advantages}) and drawbacks (\cref{sec:npe-drawbacks}) relative to traditional ELBO-based variational inference \citep[cf.][]{blei2017variational}.

\subsection{Advantages of NPE}
\label{sec:npe-advantages}
NPE has five important advantages relative to traditional ELBO-based variational inference.

\paragraph*{Versatile unbiased gradients}
Provided the model can be sampled, as it typically can, 
unbiased stochastic gradients of $\mathcal L_\mathrm{fwd}(\phi)$ can always be obtained through Equation~10 in the main text.
Hence, there is no need to use the reparameterization trick, as is often needed to minimize $\mathcal L_\mathrm{rev}(\phi)$; the ELBO averages with respect to $q_\phi$, so the averaging distribution depends on $\phi$~\citep{kingma2014auto}.
In contrast, NPE only averages with respect to $p(z, x)$, thus freeing us to use an expressive variational distribution that is not readily reparameterizable, as indeed ours is not.

\paragraph*{Likelihood-free inference}
Stochastic gradients of the NPE objective are based exclusively on samples from the joint distribution of the model;
it is unnecessary to evaluate the likelihood of complete data.
Although this likelihood is technically tractable for our model of astronomical images, it is convenient not to have to evaluate it. The popular GalSim package~\citep{rowe2015galsim}, which we use in our case studies, supports simulating images but not evaluating likelihoods.

\paragraph*{Implicit marginalization}
The NPE objective allows us to implicitly marginalize nuisance latent variables, such as random variables that describe nuisance objects, by simply excluding them from the variational distribution \cite[Theorem 2]{ambrogioni2019forward}. Note that excluding random variables from the variational distribution does not change the generative process. For our application, the synthetic images that we use to train the inference network reflect the presence of many objects that are too faint to detect individually, but affect the image in aggregate. The inference network is trained to ignore these faint objects, which implicitly marginalizes over them. 
% Marginalizing over these objects is essential to fit a variational distribution with support for only a bounded number $M$ of objects because the expected forward KL divergence is infinite unless $q(z \mid x) \gg p(z \mid x)$ for all $x \in \mathrm{support\{ }p(x)\}$, while the prior object count, which follows a Poisson distribution, has unbounded support.
This marginalization procedure is not equivalent to excluding faint objects from the generative model.
If we were to revise the generative model to exclude samples with more than $M$ objects per tile, then the generative model would be misspecified, as real astronomical images have, with some probability, regions with high object densities. 
The inference network may perform poorly for such images because they are out-of-distribution.

\paragraph*{Not underdispersed}
If the variational family includes the exact posterior, then the minimizers of both the reverse KL and the forward KL are identical (and both are identical to the posterior distribution). However, variational distributions, including ours, are essentially never flexible enough to include the exact posterior. Therefore, the minimizer selected will depend on the divergence minimized. Whereas the reverse KL tends to favor under-dispersed (``mode seeking'') variational distributions, the forward KL tends to favor over-dispersed (``mode covering'') ones \citep[Ch. 10]{bishop2006pattern}.
Typically the latter is preferable, as the costs of being overly confident in applications are greater than the costs of being conservative. For example, if the variational distribution is used downstream as the proposal for importance sampling, the forward KL bounds the effective sample size whereas the reverse KL does not~\citep{chatterjee2018sample}.

\paragraph*{Favorable optimization landscape}
The ELBO is highly nonconvex for object detection in astronomical images~\citep[Section 5]{liu2023variational} and object detection in natural images \citep{greff2019multi}.\footnote{\cite{greff2019multi} does the not explicitly discuss the convexity of the ELBO, but proposes a complex ELBO optimization scheme that presumably would have been unnecessary had the ELBO been well behaved.} Intuitively, this nonconvexity arises from the discrete nature of object detection: the gradient provides little or no signal to guide the optimizer towards objects that are entirely missed at an earlier iteration.
Therefore, optimizers can get stuck in shallow local optima regardless of the precise optimizer used---a long-standing problem in VI~\citep{blei2017variational}.

Recent work motivated by our experience with object detection has shown that minimizing the NPE objective with stochastic gradient descent is globally convergent~\citep{mcnamara2024globally}. This result is even more surprising when we consider that it holds in the amortized setting, in which the variational distribution is parameterized by a neural network, leading to a highly nonconvex optimization problem in the space of the neural network weights. 
The key to this result is analyzing the NPE objective using the neural tangent kernel~\citep{jacot2018neural}. (Recall that algorithmically, NPE objective minimization amounts to simulating data and training an inference network to solve the inverse problem, effectively a supervised problem with an arbitrarily large training set.)  This approach requires some assumptions, chiefly that the inference network architecture is infinitely wide, but \cite{mcnamara2024globally} presents compelling empirical evidence that the asymptotic regime is relevant to practice. No comparable guarantees are available for the ELBO. 

\subsection{Drawbacks of NPE}
\label{sec:npe-drawbacks}
We see four potential issues with NPE. Only the first (``overdispersion'') is widely appreciated. The second has only been reported recently. The third and fourth are novel insights.

\paragraph*{Overdispersion} % Overdispersion causes sensitivity to parameterization
Posterior approximations that minimize forward KL divergence tend to be overdispersed \citep[Ch. 10]{bishop2006pattern}.
% Intuitively, this is because the forward KL, which is an expectation over data sampled from the model, directly encourages the variational distribution to assign high probability to simulated data, but only indirectly penalizes the overassignment of probability mass to data that are unlikely or even impossible (i.e., outside of the support) under the target distribution.
While overdispersion is often preferable to underdispersion (see \cref{sec:npe-advantages}), it is potentially problematic nevertheless.
An overdispersed posterior approximation assigns excessive probability mass to regions that have low or even no mass under the exact posterior distribution.
Therefore, if the samples from the posterior approximation are used during inference, the occasional sample may be impossible, which could invalidate analysis that is sensitive to such outliers.

\paragraph*{Vulnerability to model misspecification}
Some degree of model misspecification is inevitable.
For example, \citet{regier2019approximate} observed that the standard bulge-and-disk galaxy model, which we also use here, does appear to misfit some galaxies, despite the homogenizing effect of convolution with the point spread function for ground-based images.
This is not surprising given the diversity of galaxy shapes, especially those of irregular galaxies, and it is not inherently problematic: imperfect models can be useful.

Traditional ELBO-based methods target an objective that is the sum of two terms: the reconstruction error, $\mathbb E_{q(z \mid x)}\, \log p(x \mid z)$, and an entropic regularizer, $D_{\mathrm{KL}}( q(z \mid x), p(z) )$. The reconstruction error term strongly favors posterior approximations that assign high likelihood to the data, regardless of model misspecification. For example, for an image containing a bright galaxy that is highly irregular, the reconstruction error term will favor posterior approximations that assign high probability to modeled galaxies that resemble the real one, with the remaining misfit ``soaked up'' by the background.

The same cannot be said for NPE in the presence of model misspecification. With NPE, the inference network is trained exclusively with synthetic data from the model. Only at prediction time does it encounter real data. At prediction time, even slight deviations between real data and the synthetic training data could cause the inference network to perform poorly: neural networks are known to be highly sensitive to distribution shift, particularly if the support of the data model differs from the support of the actual data. There are methods of adding noise to the training data to make networks trained with them more robust to domain shift. These methods are popular in fields such as deep learning for medical imaging and have recently been developed for NPE by~\citet{ward2022robust}.

\paragraph*{Malignant overfitting}
Neural networks that are trained with stochastic gradient descent to achieve zero training error---and thus to thoroughly overfit the training data---nevertheless frequently perform well in terms of test error.
This paradox, termed ``benign overfitting,'' has been widely observed and studied, typically in the context of mean squared error (MSE) loss \citep{zhu2023benign}. 

In neural posterior estimation, the loss function is not equivalent to MSE (except in the extremely rare case in which the variational family is made up of Gaussian distributions with fixed isotropic covariance).
Instead, the inference network infers multiple moments of the posterior distribution.
Overfitting is far from benign in this setting: once the network has memorized the training data, the predicted uncertainty of predictions goes to zero. This unbounded overconfidence leads to arbitrarily large loss for held-out data.

We term this potential problem ``malignant overfitting.'' For this work, because our data simulation procedure is relatively inexpensive and could in principle be performed on-the-fly during training (so that no training data point is used more than once), malignant overfitting may not be a major problem. It could become so, however, if we switch to a more computationally expensive simulator, such as PhoSim, which performs ray tracing~\citep{peterson2021phosim}.

\paragraph*{Exposure bias}
The term ``exposure bias'' was coined in the context of recurrent neural networks (RNNs) for language generation to describe the observation that RNNs trained with real text data to predict the next word in a sequence (``teacher forcing mode'') were not necessarily good at generating new text sequences (``free-running mode'') \citep[cf.][]{schmidt2019generalization,wang2020exposure,he2021exposure}. After generating a certain number of words, synthetic text sequences differ slightly from real (human-generated) text sequences. Once these small differences manifest, an RNN trained exclusively with real data and tasked with generating the rest of the words in a sequence is operating out-of-distribution: the conditioning data are unlike any that it has been trained with. The discrepancies thus begin to compound, often resulting in highly unrealistic text generation.

A similar dynamic can occur when fitting a hierarchical variational approximation with NPE. 
During training, which is based on complete data generated by ancestral sampling, the true data-generating realizations of the latent variables are known.
The variational distribution is fitted with this complete data to infer child latent variables from the values of their parents.
After training, the variational distribution can be a good predictor, but it is typically imperfect;  neural networks are not perfectly flexible, and we restrict the probabilistic form of the variational approximation.
Therefore, at inference time, to sample from the variational distribution, the parent node is sampled conditional on the data, and the child node is sampled conditional on this sample of the parent, whose distribution typically differs slightly from the prior. This slight difference can potentially have a major effect on the inference network, which at inference time is operating out of distribution.

To be clear, this is not an issue of model misspecification or overfitting. The validation error can continue to decrease as the variational distribution is fitted to make progressively better inferences about child latent variables, while exposure bias becomes more severe.

Remedies for exposure bias in RNNs for text generation may be applicable to NPE.
Scheduled sampling \citep{bengio2015scheduled} is a popular remedy that, with some frequency during training, replaces the true context information with a sample of the distribution being fitted. In this way, at inference time, the inference network is not run in an out-of-distribution context.
Unfortunately, as pointed out by \cite{huszar2015not}, scheduled sampling is not a consistent estimation strategy: the exact posterior is typically not the minimizer of the resulting objective.
\cite{huszar2015not} posits that the empirical success of scheduled sampling relies on incomplete optimization (e.g., getting stuck in a local optimum), which is an unsatisfying basis for inference.

\cite{lamb2016professor} propose to augment the likelihood-based training objective with an adversarial penalty to counteract exposure bias in RNNs. This penalty discourages representations of contextual information that distinguish training data from samples. This approach is appealing for NPE too, as it has the potential to correct for some forms of overdispersion (discussed above) by penalizing samples with more probability mass under the aggregate posterior than the prior.

\newpage

\section{Neural Network Architecture}

The inference network described in Section 4.3 of the main text processes astronomical images through a series of transformations to detect and characterize sources in a tile-based manner. Here we detail the network architecture of its three components: the image backbone, the neighborhood network, and the detection head.

\subsection{The image backbone}
The image backbone, shown in \cref{fig:image_backbone}, takes as input a minibatch of images, $X \in \mathbb{R}^{B \times 1 \times H \times W}$, where $B$ is the batch size, $1$ indicates that the image has a single channel (the r band), and $H$ and $W$ are the height (rows) and width (columns) of the images, in pixels. 

\begin{figure}[b]
    \centering
    \includegraphics[width=\linewidth]{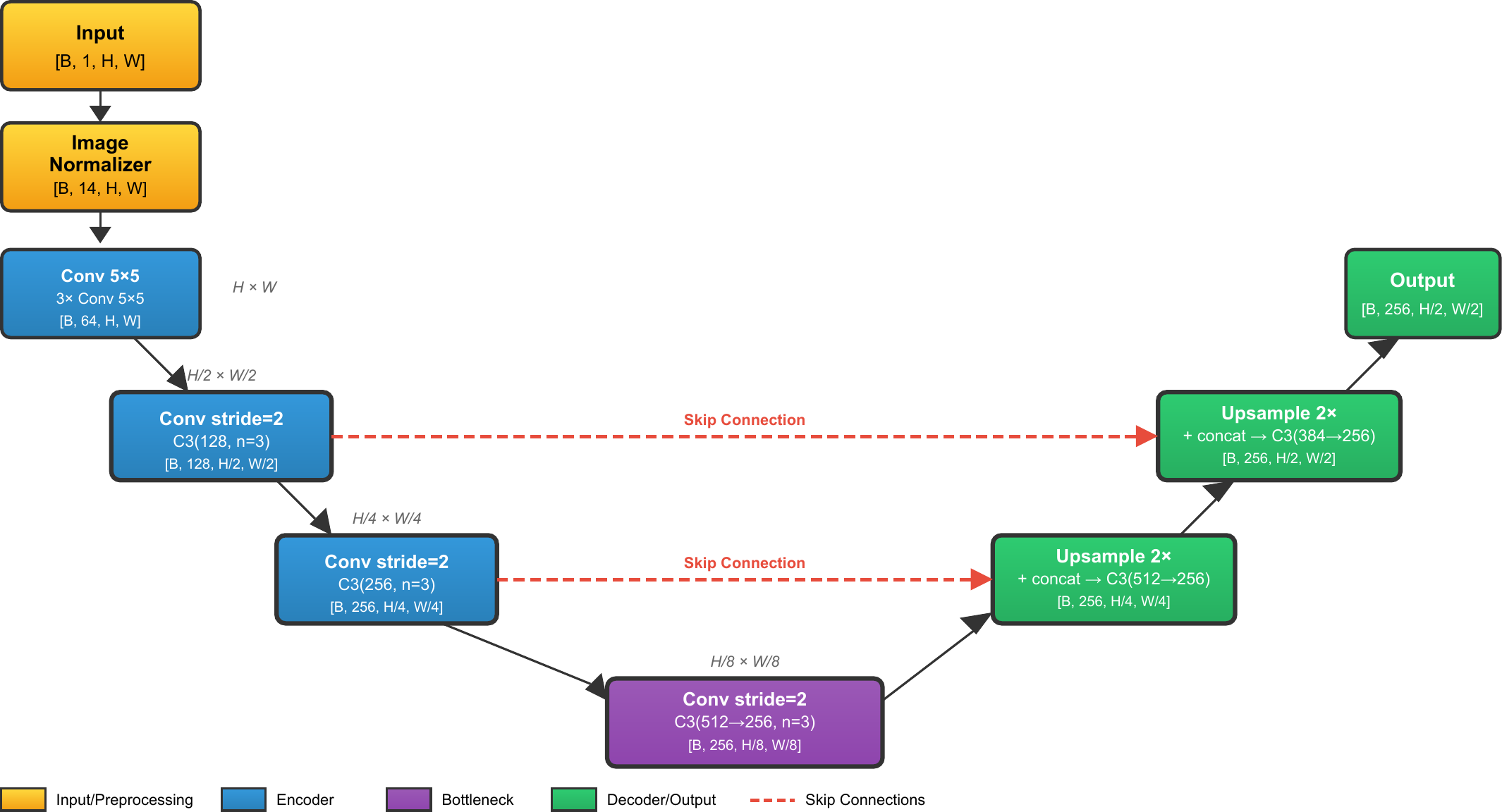}
    \caption{Image backbone architecture.}
    \label{fig:image_backbone}
\end{figure}

First, an image normalizer preprocesses the astronomical images using an inverse hyperbolic sine (asinh) transformation, which provides compression of large dynamic ranges while preserving fine details in low-intensity regions---a particularly useful property for astronomical imaging where pixel values can span several orders of magnitude from faint background to bright stellar sources. The normalizer produces 14 representations of each input image by applying 14 different elementwise affine transformations (scaling and shifting) to the pixels, followed by the asinh function. The amount of scaling and shifting for each representation is determined by quantiles computed from sky-subtracted input images.

Next comes the U-net. A series of 2D convolutional blocks process and downsample the feature maps while successively increasing the number of channels. Then, a series of paired upsampling operations run, with input coming both from the previous upsampling block and a skip connection to the corresponding downsampling block.

\subsection{The neighborhood network}
The neighborhood network, shown in \cref{fig:neighbor_network}, processes two distinct types of contextual information through separate pathways that are then combined to produce the final neighborhood context representation.

The neighboring-tile context pathway processes objects from tiles with ranks less than the current rank ($k=1,...,K$). The input consists of object properties and a rank-based tile mask that identifies which spatial positions contain objects from qualifying neighboring tiles. These inputs are fed into a ConvBlock with 3×3 kernels, producing a 128-channel feature map. This is followed by a sequence of three cross-stage partial blocks \citep{redmon2016you} that maintain the same dimensional output, and finally a ConvBlock with 1×1 kernels that preserves the dimensions.

\begin{figure}[t]
    \centering
    \includegraphics[width=0.9\linewidth]{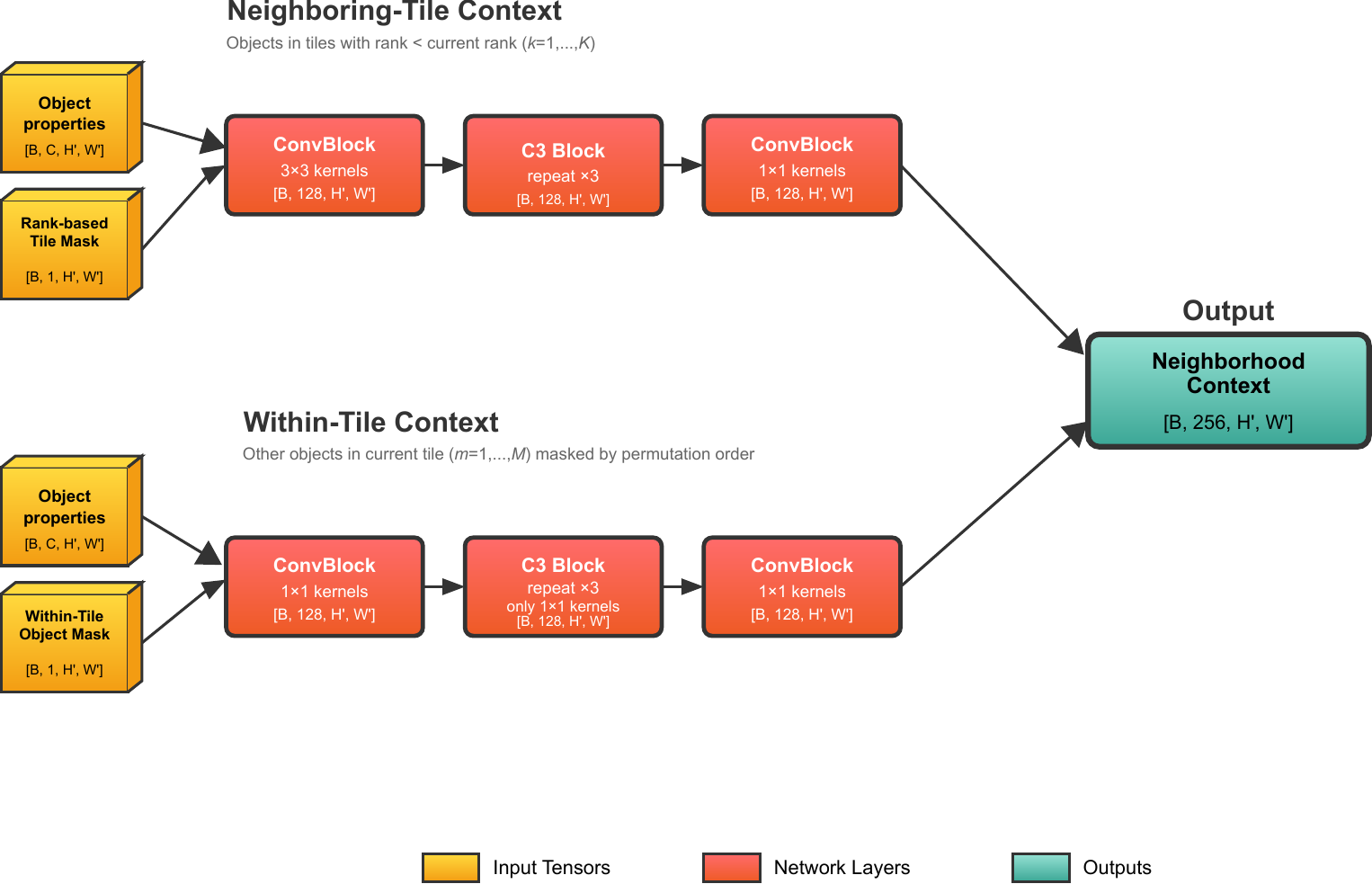}
    \caption{Neighborhood network architecture.}
    \label{fig:neighbor_network}
\end{figure}

The within-tile context pathway handles other objects within the current tile ($m=1,...,M$) that are masked according to the permutation order. Similar to the neighboring-tile pathway, it takes object properties and a within-tile object mask. However, this pathway uses only 1×1 kernels throughout: a ConvBlock with 1×1 kernels, followed by a C3 Block (repeated 3 times, using only 1×1 kernels), and a final ConvBlock with 1×1 kernels.

The features from both pathways are concatenated to produce the final neighborhood context tensor with 256 channels.
The architectural distinction between the two pathways reflects their different spatial scopes: the neighboring-tile pathway uses 3×3 kernels to capture spatial relationships across tile boundaries, while the within-tile pathway restricts itself to 1×1 kernels to respect the autoregressive ordering constraint within the current tile.

\subsection{The detection head}
The detection head, shown in \cref{fig:detection_head}, combines neighborhood context and image backbone features by concatenating them along the channel dimension. The concatenated features are then processed through a sequence of three convolutional stages using exclusively 1×1 kernels.
The first stage is a ConvBlock with 1×1 kernels that omits GroupNorm, to prevent information sharing across spatial positions. This is followed by a C3 Block that repeats three times, again using only 1×1 kernels and preserving the same dimensional structure. The final stage applies a 1×1 convolution with bias and includes a dimension permutation operation.
The output has dimension $[B, H', W', n\_params]$, where $n\_params$ represents the distributional parameters for potential objects at each spatial location. The exclusive use of 1×1 kernels ensures that parameters are computed independently for each spatial position, maintaining the autoregressive property while leveraging the  information from both input sources.
\begin{figure}[H]
    \centering
    \includegraphics[width=0.9\linewidth]{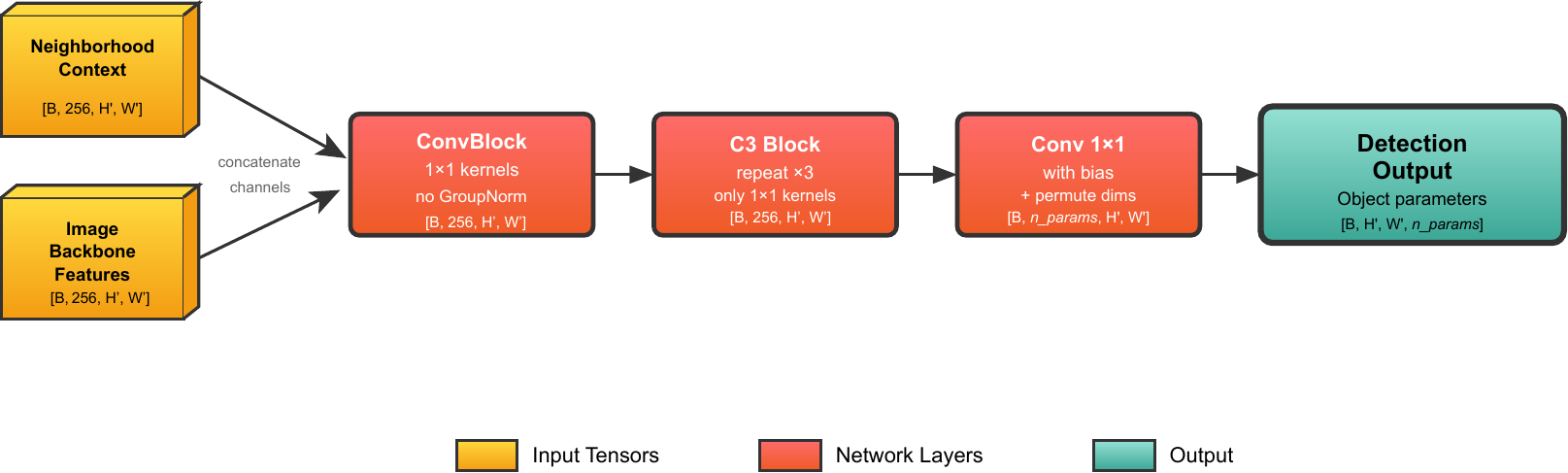}
    \caption{Detection head architecture.}
    \label{fig:detection_head}
\end{figure}

\newpage
\section{Overly rich conditioning information for Case Study 2}
\label{rich_conditioning}

\begin{figure}[H]
    \centering
    \includegraphics[width=0.48\linewidth]{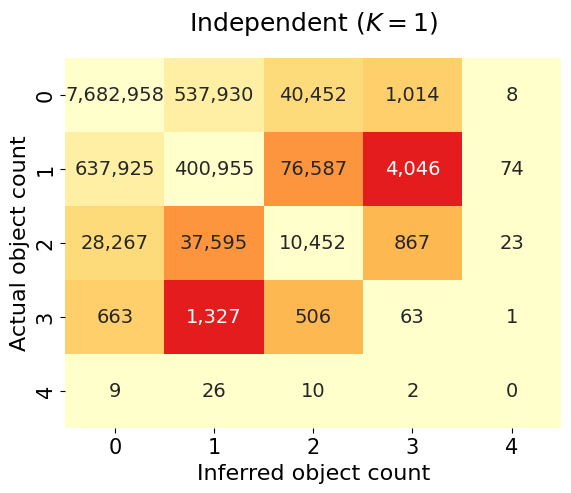}
    \hspace{10px}
    \includegraphics[width=0.48\linewidth]{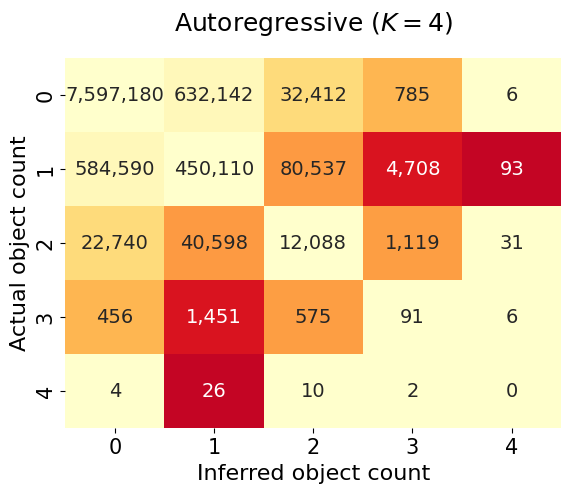}
    \caption{Confusion matrices for samples of catalogs from the BLISS variational distribution \textbf{with rich conditioning information} for synthetic images that mimic the SDSS image of the M2 globular cluster. For each catalog, whether inferred or actual, object count is the number of objects in a thin region along tile boundaries.
    Each cell's color indicates the difference between the entries of the corresponding transpose pair as a proportion of the smaller of the two entries (for transpose pairs with at least 100 observations in each cell).
    }
    \label{fig:m2_maximalist_confusion_matrices}
\end{figure}

\end{document}